\documentclass[letterpaper, 10pt, conference]{ieeeconf}      

\IEEEoverridecommandlockouts                              

\usepackage{tikz}
\usepackage{import}
\usepackage{graphics}
\usepackage{graphicx}
\graphicspath{{figures/}}
\usepackage[hidelinks]{hyperref}
\usepackage{flushend}

\usepackage{amssymb,amsmath,mathrsfs,siunitx,empheq}
\usepackage{mathtools}

\usepackage{amsthm}

\def\vec{\mathaccent "017E\relax }

\usepackage{color,xcolor}
\definecolor{p2color}{HTML}{D95319}
\definecolor{pscolor}{HTML}{77AC30}
\definecolor{hpscolor}{HTML}{7E2F8E}

\usepackage[font=footnotesize,caption=false,format=hang,margin=1em]{subfig}
\usepackage{algorithm}
\usepackage{algpseudocode} 
\usepackage{nicematrix}
\usepackage{arydshln}
\usepackage{dashrule}

\usepackage[normalem]{ulem}

\usepackage{multirow}
\usepackage{booktabs}

\let\labelindent\relax
\usepackage{enumitem}
\renewcommand{\theenumi}{\rm{(\roman{enumi})}}

\newcommand{\btitle}[1]{\mbox{}{\bf{(#1).}}}
\newtheorem{theorem}{Theorem}[section]

\newtheorem{proposition}[theorem]{Proposition}

\newcommand{\real}{\ensuremath{\mathbb{R}}}
\newcommand{\realpos}{\ensuremath{\mathbb{R}_{>0}}}
\newcommand{\realnonneg}{\ensuremath{\mathbb{R}_{\ge 0}}}
\newcommand{\naturals}{\ensuremath{\mathbb{N}}}

\newcommand{\intgnonneg}{\ensuremath{\mathbb{Z}_{\ge 0}}}

\newcommand{\SO}{{\rm{SO}}(3)}

\newcommand{\mb}[1]{\mathbf{ #1 }}

\newcommand{\col}{\mathrm{col}}

\newcommand{\diag}{\mathrm{diag}}

\newcommand{\paren}[1]{\left(#1\right)}

\newcommand{\braces}[1]{\left\{#1\right\}}

\newcommand{\s}{\scriptscriptstyle}

\newcommand{\sigt}{\sigma(t)}

\newcommand{\Ac}{\mathcal{A}}
\newcommand{\Bc}{\mathcal{B}}

\newcommand{\Dc}{\mathcal{D}}
\newcommand{\Ec}{\mathcal{E}}
\newcommand{\Fc}{\mathcal{F}}
\newcommand{\Gc}{\mathcal{G}}
\newcommand{\Hc}{\mathcal{H}}
\newcommand{\Ic}{\mathcal{I}}

\newcommand{\Mc}{\mathcal{M}}
\newcommand{\Nc}{\mathcal{N}}
\newcommand{\Oc}{\mathcal{O}}

\newcommand{\Qc}{\mathcal{Q}}

\newcommand{\Vc}{\mathcal{V}}

\newcommand{\px}{{p}^{x}}
\newcommand{\py}{{p}^{\smash y}}
\newcommand{\pz}{{p}^{z}}

\newcommand{\smsh}[1]{\vphantom{{\bar{#1}}}#1}

\newcommand{\pxi}{\px_{\smsh{i}}}
\newcommand{\pyi}{\py_{\smsh{i}}}
\newcommand{\pzi}{\pz_{\smsh{i}}}

\newcommand{\ddotpxi}{\ddot{p}^{x}_{\smsh{i}}}
\newcommand{\ddotpyi}{\ddot{p}^{ y}_{\smsh{i}}}
\newcommand{\ddotpzi}{\ddot{p}^{z}_{\smsh{i}}}

\newcommand{\dotpxi}{\dot{p}^{x}_{\smsh{i}}}

\newcommand{\ddotpxip}{\ddot{p}^{x'}_{\smsh{i}}}
\newcommand{\ddotpyip}{\ddot{p}^{ y'}_{\smsh{i}}}
\newcommand{\ddotpzip}{\ddot{p}^{z'}_{\smsh{i}}}

\newcommand{\xr}{{\rm{x}}}

\newcommand{\pr}{{\rm{p}}}
\newcommand{\vr}{{\rm{v}}}

\newcommand\nullH{\Hc^{\raisebox{0pt}{\scriptsize $\mathfrak{0}$}}}
\newcommand\alterH{\Hc^{\raisebox{0pt}{\scriptsize $\mathfrak{1}$}}}

\newcommand{\bxtar}{\mb{x}^{\boldsymbol{*}}}
\newcommand{\rxtar}{{\rm x}^{\boldsymbol{*}}}

\newcommand{\pstar}{\boldsymbol{p}^{\boldsymbol{*}}}

\newcommand{\ua}{\mb{u}_{\rm a}}
\newcommand{\us}{\mb{u}_{\rm s}}

\newcommand{\uar}{{\rm u}_{\rm a}}
\newcommand{\usr}{{\rm u}_{\rm s}}

\newcommand{\pb}{\boldsymbol{p}}
\newcommand{\vb}{\boldsymbol{v}}
\newcommand{\ub}{\boldsymbol{u}}

\newcommand{\efrak}{\mathfrak{e}}
\newcommand{\bfrak}{\mathfrak{b}}

\newcommand{\gframe}{\{\boldsymbol{\Ic}\}}
\newcommand{\bframe}{\{\boldsymbol{\Bc}^i\}}
\newcommand{\iframe}{\{\boldsymbol{\Ic'}\}}

\newcommand{\gframeAxes}{\braces{\vec{\efrak}_x, \vec{\efrak}_y, \vec{\efrak}_z}}
\newcommand{\bframeAxes}{\{\vec{\bfrak}^i_1, \vec{\bfrak}^i_2, \vec{\bfrak}^i_3\}}
\newcommand{\iframeAxes}{\braces{\vec{\efrak}_x\vphantom{.}\!', \vec{\efrak}_y\vphantom{.}\!', \vec{\efrak}_z\vphantom{.}\!'}}

\newcommand{\adj}{\mathsf{A}}
\newcommand{\lap}{\mathsf{L}}

\newif\ifshort
\shorttrue

\title{\LARGE \bf
Detection of Stealthy Adversaries for Networked Unmanned Aerial Vehicles
}

\author{Mohammad Bahrami and Hamidreza Jafarnejadsani
\thanks{Mohammad Bahrami and Hamidreza Jafarnejadsani are with the Department of Mechanical Engineering, Stevens Institute of Technology, Hoboken, NJ 07030, USA,
        {\tt\small \{mbahrami,hjafarne\}@stevens.edu}.}
}

\begin{document}

\maketitle
\thispagestyle{empty}
\pagestyle{empty}

\begin{tikzpicture}[overlay, remember picture]
	\path (current page.north) ++(0.0,-1.0) node[draw = black] {\small This paper has been accepted for publication in the proceedings of 2022 International Conference on Unmanned Aircraft Systems (ICUAS)};
\end{tikzpicture}
\vspace{-0.3cm}

\begin{abstract}
A network of unmanned aerial vehicles (UAVs) provides distributed coverage, reconfigurability, and maneuverability in performing complex cooperative tasks. However, it relies on wireless communications that can be susceptible to cyber adversaries and intrusions, disrupting the entire network's operation. This paper develops model-based centralized and decentralized observer techniques for detecting a class of stealthy intrusions, namely zero-dynamics and covert attacks, on networked UAVs in formation control settings.  The centralized observer that runs in a control center leverages switching in the UAVs' communication topology for attack detection, and the decentralized observers, implemented onboard each UAV in the network, use the model of networked UAVs and locally available measurements. Experimental results are provided to show the effectiveness of the proposed detection schemes in different case studies.
\end{abstract}
\section*{Supplementary Material}
\small
Video: {\footnotesize \url{https://youtu.be/lVT_muezKLU}}

Code: {\footnotesize \url{https://github.com/SASLabStevens/TelloSwarm}}
\normalsize
\section{Introduction}\label{Sec:introduction}
Teams of autonomous robots, particularly Unmanned Aerial Vehicles (UAVs), are of interest as their cooperation allows for distributed coverage,  reconfigurability, and mobility in a wide range of applications such as search and rescue missions, surveillance, wildfire monitoring, and delivery. 
Networked UAVs rely on information exchange over a wireless communication network to coordinate and cooperate. 
However, numerous studies have shown the vulnerabilities of wireless networks to spoofing attacks, malicious intrusions, and Denial of Service (DoS), which raises security concerns in safety- and time-critical applications. 
This paper addresses the problem of detecting a class of \emph{stealthy} attacks on the control system of networked UAVs in a setting where UAVs cooperate to achieve a particular formation, and their communication network is subject to topology switching.
\begin{figure}[t]
    \centering
    \includegraphics[width=.8\linewidth]{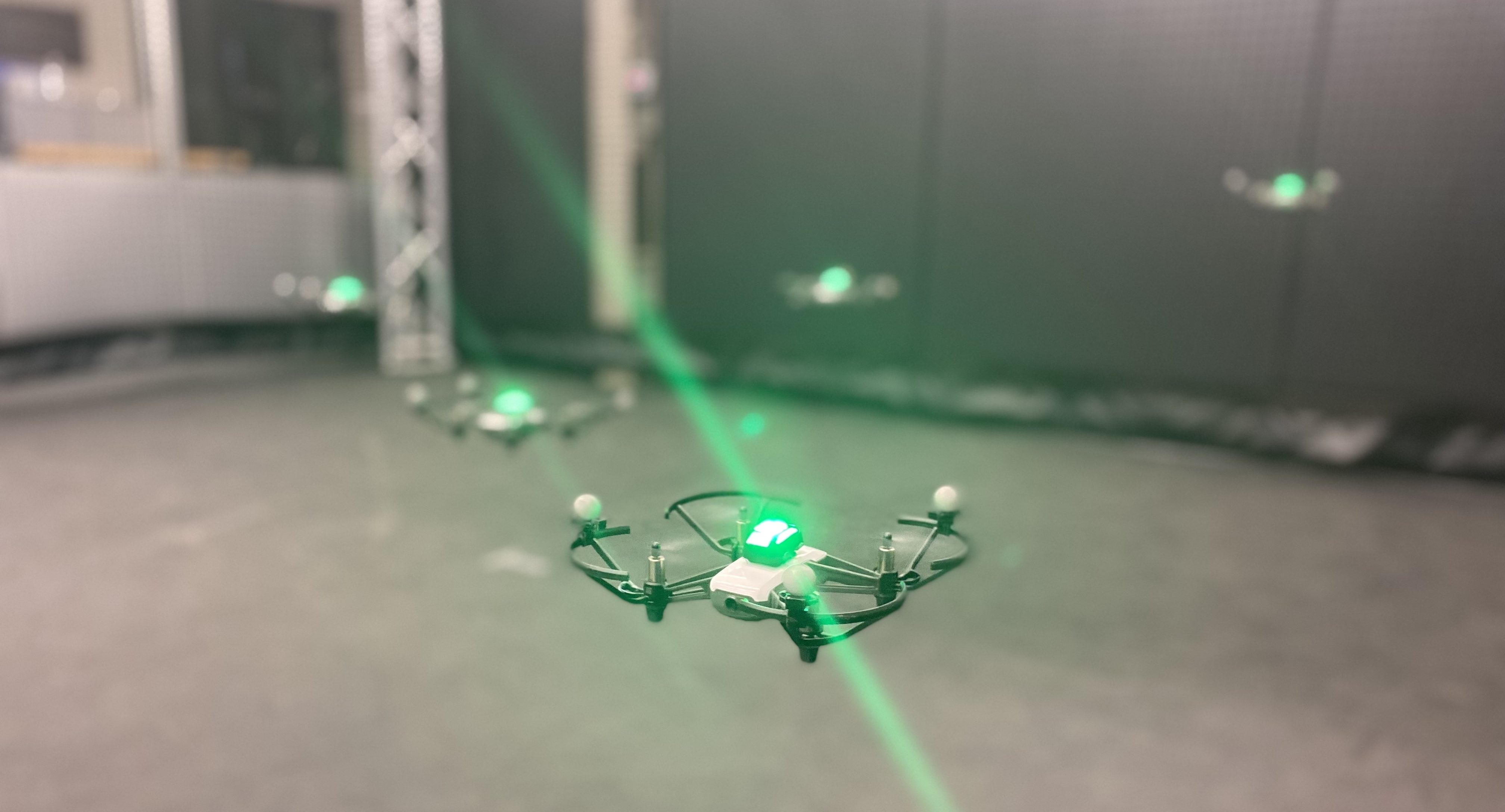}
    \caption{ 
    The experimental setup for multi-UAV formation control using a motion capture system for positioning.
    } 
    \label{fig:tello}
\end{figure}

\noindent
\subsection{Related work} 

\noindent
\emph{Multi-UAV cooperation.} Cooperation and coordination of a network of robots (UAVs) in the forms of swarming and formation control have been extensively studied 
\ifshort
\cite{chung2018survey,oh2015survey}.
\else
\cite{chung2018survey,oh2015survey,liu2018survey}.
\fi
Formation control has been adopted in performing collaborative tasks such as coverage control for sensor networks \cite{cortes2004coverage}, forest firefighting \cite{harikumar2018multi}, and sensor planning for precision agriculture \cite{tokekar2016sensor}.
A team of UAVs can achieve formation by exchanging their spatial local information (e.g. position and velocity states) and in different settings such as leader-follower, leaderless consensus-based, and virtual structure approaches \cite{oh2015survey}.
This paper focuses on consensus-based formation control of a team of UAVs coordinating their relative positions over a switching communication network.  

\noindent
\emph{Adversary and cyber-attack detection for UAVs.} 
The susceptibility of networked UAVs to cyber adversaries and attacks has been reported in
\ifshort
\cite{manesh2019cyber,yaacoub2020security}.
\else
\cite{manesh2019cyber,yaacoub2020security,chriki2019fanet}.
\fi
There are also different defense mechanisms against attacks and adversaries. 
Intrusion Detection Systems (IDS) that operate often based on the statistical characteristics of transmitted data can significantly limit possible attacks, thereby improving the network security \cite{choudhary2018intrusion}.
However, IDS' invariance to the system's dynamics renders them incapable of detecting stealthy attacks devised based on networked UAVs' dynamical characteristics. 
Examples of such attacks are zero-dynamics attack \cite{mao2020novel}, covert attack \cite{smith2015covert}, and replay attack \cite{mo2009secure}.
Challenges of IDS and also model-based monitoring frameworks in detecting stealthy attacks on a single UAV have been studied in \cite{kwon2014analysis,dash2019out}.
Redundancy in hardware and software has been adopted to detect and mitigate a class of attacks on a single UAV \cite{fei2018cross}. 
Secure planning against stealthy (GPS spoofing) attack has been studied using control theoretic \cite{liu2020secure} and reinforcement learning \cite{bozkurt2021secure} approaches.
Only a few studies have considered the detection problem of stealthy attacks on networked UAVs, which are encryption and encoding for detection of false data injection (FDI) attacks \cite{liu2020fdi}, information fusion for GPS spoofing detection \cite{liang2020detection}, distributed observers for cyberattack detection and isolation \cite{negash2017distributed}, and finally learning-based methods to detect stealthy FDI attacks targeting a single UAV in cooperative localization \cite{khanapuri2022learning} and to detect anomalies in swarming drones \cite{ahn2020deep}.
However, none of these studies have considered the detection problem of stealthy attacks, namely zero-dynamics attack (ZDA) and covert attack, on coordination of networked UAVs.

\noindent
\emph{Statement of contributions.} 
We study the detection of \emph{stealthy} attacks on a team of UAVs coordinating their relative positions, over a wireless network with switching communication links, to achieve formation. Stealthy attacks are designed based on the high-level coordination model (communication topology) of the UAVs to maximize the attack's effect on the cooperation performance as well as the attack's stealthiness in monitored signals. 

The main contributions of this paper are as follows: We develop centralized and decentralized attack detection strategies against \emph{stealthy} attacks, namely zero-dynamics attack (ZDA) and covert attack, that adversely affect the coordination objective of networked UAVs while remaining stealthy in the monitoring signals. We investigate switching communication links in terms of their role in attack detection and system reconfigurability. We demonstrate an experimental implementation of the stealthy attacks on a team of small UAVs and evaluate the effectiveness of our proposed methods in the timely detection of stealthy attacks.
\section{Problem Formulation}\label{Sec:problem_formulation}
\subsection{Notations}
We use  $ \mathbb{R} $, $ \mathbb{R}_{> 0} $, $ \mathbb{R}_{\geq 0} $, $\naturals$, and $\intgnonneg$ to denote the set of reals, positive reals, non-negative reals, natural numbers, and  non-negative integers, respectively. 
$\boldsymbol{1}_n$, $ \boldsymbol{0}_n$, $ {I}_n $ and $ \boldsymbol{0}_{n \times m} $ stand for the $n$-vector of all ones, the $n$-vector of all zeros, the identity $ n $-by-$ n $ matrix, and the $ n $-by-$ m $ zero matrix, respectively\footnote{We may omit the subscripts when clear from the
context.}. 
In addition, we use $\mathfrak{e}_{i}$ to denote the $i$-th canonical vector in $\real^{n}$.
We use $\col(\cdot)$ and $\diag(\cdot)$ to denote the column and diagonal concatenation of vectors or matrices, and $\otimes$ to denote the Kronecker product. Finally,
for any set $\Fc$, $ |\Fc| $ denotes its cardinality.
\subsection{Quadrotor's dynamics}
We consider a team of $N$ homogeneous unmanned aerial vehicles (quadrotor UAVs) that cooperate to achieve a geometric shape/formation in $\real^{\rm 2}$. 
Attached to the center of mass of each quadrotor, the body frame $\bframe$ with unit axes $\bframeAxes,\ i \in \{1, \dots,  N\} =: \Vc$, whose position and orientation with respect to the inertial global frame $\gframe$ with unit vectors $\gframeAxes$ (see Figure \ref{fig:quad_coord}a) are, respectively, determined by a vector ${p}_i = \col \paren{\pxi,\pyi,\pzi} \in \gframe,\ \forall \, i \in \Vc $ and a rotation matrix $R_i(\psi_i, \phi_i, \theta_i) \in \SO $ in the special orthogonal group
with $\psi_i$,  $\phi_i$, and $\theta_i$ being the respective $z\!-\!x\!-\!y$ Euler angles.
Then, the rigid body motion of the quadrotors follows \cite{mahony2012multirotor}
\begin{subequations}\label{eq:quad_dynamics}
\begin{align}
\label{eq:quad_dynamics_translational}
    \dot{p}_i &= {v}_i,
    \ \
   & m\dot{v}_i &= -mg \vec{{\efrak}}_z+R_i {f_i},
    \\
    \label{eq:quad_dynamics_rotational}
    \dot{R}_i &= R_i \Omega^{\times}_i, 
    \ \
   & J \dot{\Omega}_i &= -\Omega_i \times J \Omega_i + \boldsymbol{\tau}_i,
\end{align}
\end{subequations}
where ${p}_i \in \real^3$ and $R_i(\psi_i, \phi_i, \theta_i) \in \SO $ are the position and orientation of the $i$-th quadrotor in the inertial frame $\gframe$, $m$ is the mass of the quadrotor, $g$ is the gravitational acceleration, and finally $f_i \in \bframe$ is the total thrust. Also, in the rotational dynamics, $ \Omega_i \in \real^3 $ is the angular velocity, $J \in \real^{3 \times 3}$ is the inertia matrix, and $\boldsymbol{\tau}_i \in \real^3$ is the total torque, all expressed in respective body-fixed frames. Finally, the notation $\Omega^{\times}_i$ denotes the skew-symmetric matrix, such that $\Omega^{\times}_i r = \Omega_i \times r$ for any vector $r \in \real^3$ and the cross product $\times$.
\begin{figure}[t]
    \centering
    \includegraphics[width=.9\linewidth]{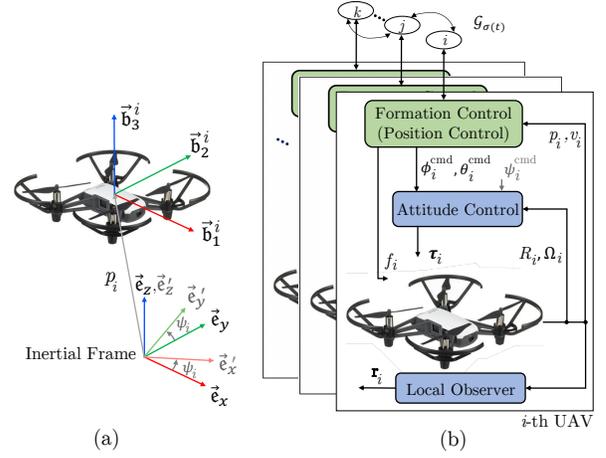}
    \caption{ 
    (a) Illustration of reference frames. (b) The coordination control architecture.
    } 
    \label{fig:quad_coord}
    \vspace*{-1ex}
\end{figure}
\subsection{Formation control}\label{Sec:formation_control}
The cooperative control of quadrotor UAVs, shown in Figure \ref{fig:quad_coord}b, follows a hierarchical structure, where at the high level, the UAVs coordinate with each other and their formation/position controller cooperatively generates the desired attitude/orientation and the desired total thrust for a low-level attitude controller.
In this paper, we focus on 2D formation in the $x\!-\!y$ plane, for which cooperative control protocols will be designed based on a linearized model of the UAVs' transnational dynamics in \eqref{eq:quad_dynamics_translational} around a hovering state and under small-angle approximations as follows \cite{mahony2012multirotor}:
\begin{subequations}\label{eq:linearized_quad_dynamics_old}
\begin{align}
    \label{eq:linearized_quad_dynamics_x_old}
    \ddotpxi &=  g (\Delta \theta_i \cos(\psi_i) +  \Delta \phi_i \sin(\psi_i)),
    \\
    \label{eq:linearized_quad_dynamics_y_old}
    \ddotpyi &= g(\Delta \theta_i \sin(\psi_i) -  \Delta \phi_i \cos(\psi_i),
    \\
    \label{eq:linearized_quad_dynamics_z_old}
    \ddotpzi &= - g + {f_i}/{m},
\end{align}
\end{subequations}
in which $\Delta \theta_i$ and $\Delta \phi_i$ denote, respectively, the deviation of pitch and roll angles of the $i$-th quadrotor from their equilibrium point $ \theta_i= \phi_i=0 $. Associated with each UAV, we define an intermediary frame $\iframe$ with unit axes $\iframeAxes$ and orientation $R_z(\psi_i)$ such that $p_i = R_z(\psi_i) p'_i$ for vectors $p_i \in \gframe$ and $p'_i \in \iframe$ (see Figure \ref{fig:quad_coord}a). Assuming all UAVs have consensus on a desired yaw angle $\psi_i = \psi^*,\ \forall \, i \in \Vc $, $\iframe$ will be the common reference frame of the UAVs in which the linearized equation of motion in \eqref{eq:linearized_quad_dynamics_old} can be represented by
\begin{subequations}\label{eq:linearized_quad_dynamics}
\begin{align}
    \label{eq:linearized_quad_dynamics_x}
    \ddotpxip &=  + g \Delta \theta_i ,
    \\
    \label{eq:linearized_quad_dynamics_y}
    \ddotpyip &= -g   \Delta \phi_i ,
    \\
    \label{eq:linearized_quad_dynamics_z}
    \ddotpzip &= - g + {f_i}/{m},
\end{align}
\end{subequations}
that shows the decoupled dynamics in the $x'_i,y'_i,z'_i$ directions\footnote{We will omit the superscript $'$ in the rest of paper for notational simplicity.}. We let the reference commands for pitch and roll angles in \eqref{eq:linearized_quad_dynamics_x}-\eqref{eq:linearized_quad_dynamics_y} be
\begin{align}\label{eq:cmd_pitch_roll}
 \theta_i^{\rm cmd} &= + u_{i}^x / g, \ \
 \phi_i^{\rm cmd}   = - u_{i}^y / g,
\end{align}
%
where $u_{i}^x$ and $u_{i}^y$ are the formation control inputs to be designed, respectively, in the $x$ and $y$ directions. It is also necessary to mention that $ \theta_i^{\rm cmd}$ and $ \phi_i^{\rm cmd}$ in \eqref{eq:cmd_pitch_roll} will be desired setpoints for each UAV's low-level (on-board) attitude controller, and that we use independent PID controllers to stabilize the altitude of quadrotors ($z_i$-dynamics in \eqref{eq:linearized_quad_dynamics_z_old}) around a desired hovering point. Therefore, the altitude dynamics in \eqref{eq:linearized_quad_dynamics_z_old} and the rotational dynamics in \eqref{eq:quad_dynamics_rotational} are dropped from the high-level state space of networked UAVs and the reduced-order planar dynamics is obtained by
substituting \eqref{eq:cmd_pitch_roll} for $\Delta \theta_i$ and $\Delta \phi_i$ in \eqref{eq:linearized_quad_dynamics_x} and \eqref{eq:linearized_quad_dynamics_y} as follows:
\begin{align}\label{eq:ol_sys}
\Sigma_{i}:
\left\{
\begin{array}{l}
\dot{\pb}_i(t) = \vb_i(t)\\ 
\dot{\vb}_i(t) = \ub_{i}(t) 
\end{array},
\right. 
\ \  i \in \mathcal{V} = \{1,\dots, N\} , 
\end{align}
in which $ \boldsymbol{p}_i(t) \coloneqq \col( \pxi, \pyi) \in \real^2,$ and $ \vb_i(t) \coloneqq \col(\dotpxi, \dotpxi) \in \real^{2} $ are the stacked positions and velocities in the $x$ and $y$ directions, and $\ub_{i}(t) \coloneqq \col( u_{i}^x, u_{i}^y)  \in \real^{2}  $ denotes their corresponding control input for each UAV. 

\noindent
\emph{{Desired formation reference}.} 
We define a desired configuration (formation shape) by specifying a set of $N$ desired setpoints $\pstar_{1}, \pstar_2, \dots, \pstar_{\s N}$ in $ \real^{2}$ that form the desired relative positions\footnote{In the context of formation control, these reference states are called \emph{formation states} \cite{ren2008distributed} or \emph{shape vectors}  \cite{turpin2012decentralized} depending on the design methods and their underlying assumptions.} $\{ \pstar_{ij}= \pstar_{i}-\pstar_{j}  \in \real^{2} \mid \forall \, i,j \in \Vc ,\ i\neq j\}$, all expressed in the UAVs' common frame. The formation references are transmitted to the UAVs from a ground control center.
We follow the consensus-based formation settings \cite{ren2008distributed} where the UAVs coordinate their relative positions to reach the desired relative positions $\pstar_{ij}$'s, which is formulated as
\begin{subequations}\label{eq:formation_consensus}
\begin{align} \label{eq:formation_consensus_pos}
\lim \limits_{t\rightarrow{\infty}} \left|\pb_i(t)-\pb_j(t) - \pstar_{ij} \right| &= \boldsymbol{0},
  &\forall \; i,j \in \mathcal{V},
\\ \label{eq:formation_consensus_vel}
\lim \limits_{t\rightarrow{\infty}} \left|\vb_i(t) \right| &= \boldsymbol{0}, 
  &\forall \; i \in \mathcal{V}.
\end{align}
\end{subequations}

\noindent
\emph{{Inter-UAV communication.}} 
We model the switching inter-UAV communication by an undirected graph $ \Gc_{\sigt}=(\Vc,\Ec_{\sigt}) $, where the vertex set $\Vc=\{1,\dots, N\}$ represents the index set of $N$ UAVs in \eqref{eq:quad_dynamics} (with their respective reduced models in \eqref{eq:ol_sys}), and the edge set $\Ec_{\sigt} \subseteq \Vc \times \Vc$ represents the communication links such that an edge $(i,j)\in \Ec_{\sigt}$ implies information exchange between the $i$-th and $j$-th UAV in a given active mode determined by the right-continuous switching signal $\sigt : \mathbb{R}_{\geq 0}\rightarrow{\Qc:=\{1,2,\dots,q\}}, \, q \in \naturals,$ at time $t$, with $\Qc$ being the finite index set of possible communication graphs.
The UAVs' interaction is further encoded into a symmetric adjacency matrix $\adj_{\sigt} \coloneqq  [a^{\sigt}_{ij}]\in \mathbb{R}^{ N \times N}_{\geq0}$ such that $a^{\sigt}_{ij} = a^{\sigt}_{ji} = 1 $ if an edge $ (i,j) \in \Ec_{\sigt}$, and $a^{\sigt}_{ij}= a^{\sigt}_{ji} = 0$, otherwise.
Also, the set of the neighbors of the $ i $-th node (UAV) in any active mode $\sigma(t)$ is defined as $ \Nc_{\sigma(t)}^i=\{j\in \mathcal{V} \mid (i,j)\in \Ec_{\sigt} \} $.

Throughout this paper, we assume the inter-UAV's communication graphs $\Gc_{\sigt}$'s are connected in all modes. 

To meet the formation constraints in \eqref{eq:formation_consensus}, we use the following consensus-based distributed control protocol:
\noindent
\begin{align}\label{eq:ctrl_proto}
    \ub_i &=  \ub_{{\rm{n}}_i} + {\ub}_{{\rm a}_i}, \quad i \in \Vc,
    \\ 
    \label{eq:ctrl_proto_normal}
    \ub_{{\rm{n}}_i} &= 
    -\alpha
    \hspace{-1ex}
    \sum_{j \in \Nc^{i}_{\sigma(t)}}
    \hspace{-1ex} 
    a^{\sigma(t)}_{ij}(\pb_i-\pb_j -\pstar_{ij} ) 
    - \gamma \vb_i 
    \vspace*{-1ex}
\end{align}
%
where $\ub_{{\rm{n}}_i} \in \real^{2} $ denotes the nominal control input with $ a^{\sigma(t)}_{ij} $ being the entry of the symmetric adjacency matrix associated with the UAVs' switching communication graph $ \mathcal{G}_{\sigma(t)} $. Also, $\alpha \in \realpos $ and $\gamma \in \realpos $ are the control gains. $\ub_{{\rm a}_i} \in \real^{2} $ is the vector-valued malicious signal injected in the control channel of the $i$-th UAV. 

We let an unknown subset $\Ac =\{i_1, i_2, \dots \} \subset \Vc$ denote the set of UAVs subject to attack  ${\ub}_{{\rm a}_i} \neq \boldsymbol{0} $, which we refer to as compromised UAVs, and we refer to the rest of UAVs $ \Sigma_i\text{'s with}\; {\ub}_{{\rm a}_i} = \boldsymbol{0},\; \forall \, i \in \Vc \setminus {\Ac} $, in \eqref{eq:ol_sys} as uncompromised UAVs.

\noindent
\emph{{Network-level dynamics.}} 
Given \eqref{eq:ol_sys}, \eqref{eq:ctrl_proto} and \eqref{eq:ctrl_proto_normal}, the dynamics of the networked UAVs can be represented by
\begin{align}\label{eq:cl_sys}
\Sigma:
\dot{\mb{x}}
    =
    \mb{A}_{\sigma(t)}
    \mb{x}
    +
    \mb{B}^{\boldsymbol{F}}_{\sigma(t)} \bxtar
    + 
  \mb{B}_{\! \Ac} \ua, \ \ \mb{x}(t_0)=\mb{x}_0, 
\end{align}
in which, the system states and matrices are given by
\begin{subequations}\label{eq:sys_matrices_states}
\begin{align}
\mb{x}(t) &= \col \paren{ \pb_1,\dots,\pb_{\s N},\vb_1,\dots,\vb_{\s N}} \in \real^{4N}, 
\\
{\bxtar} &= \col \paren{ \pstar_{1}, \dots, \pstar_{\s N}, \boldsymbol{0}_{2} , \dots , \boldsymbol{0}_{2}} \in \real^{4N},
\\
\mb{A}_{\sigma(t)}
   &=
    A_{\sigt} \!\otimes\! I_2, 
    \, 
    \mb{B}^{\boldsymbol{F}}_{\sigma(t)} \!=\! - \mb{A}_{\sigma(t)},
    \,
    \mb{B}_{\! \Ac} \!=\! B_{\! \Ac} \!\otimes\! I_2,
    \\ \label{eq:A_B_matrices}
      A_{\sigt} &=
    \begin{bmatrix}
    0_{N \times N} & I_{N} 
    \\
    -\alpha \lap_{\sigma(t)} &  -\gamma I_{N}
    \end{bmatrix},
     \; 
     B_{\! \Ac} = \begin{bmatrix}
    0_{} \\
    {\rm B}_{\! \Ac} 
    \end{bmatrix},
   \\ \label{eq:sys_matrices__attack_sig}
     {\rm B}_{\! \Ac} &= \big[\efrak_{i_1}\; \efrak_{i_2}\; \dots\; \efrak_{i_{\s |\Ac|}}\big], 
     \ \ 
     \ua = \col \paren{ {\ub}_{{\rm a}_{i}} }_{i \in \Ac},
\end{align}
\end{subequations}
where $ \lap_{\sigma(t)} $ is the Laplacian matrix of graph $\Gc_{\sigt}$, encoding the inter-UAVs' communication links, defined as $\lap_{\sigt} \coloneqq  [l^{\sigt}_{ij}] \in \real^{N \times N}$ with $ l^{\sigt}_{ii} = \sum_{j \neq i}^{}a^{\sigt}_{ij}$ and $ l^{\sigt}_{ij} = -a^{\sigt}_{ij}$ if $i \neq j$.
$\mathfrak{e}_i$ is the $i$-th vector of the canonical basis in $\real^N$ corresponding to the $i$-th UAV compromised by attack ${\ub}_{{\rm a}_i},\ i \in \Ac$.

\noindent
\emph{{Network-level measurements.}} We define the system measurements $\mb{y}$ to be composed of the position of a set of UAVs $\Mc_{\rm p} =\{{\rm p}_1, {\rm p}_2, \dots\} \subset \Vc$ and/or the velocity of a set of UAVs $\Mc_{\rm v} =\{{\rm v}_1, {\rm v}_2, \dots\} \subset \Vc$ that are transmitted to a ground control center for monitoring. More precisely,
\begin{subequations}\label{eq:measurments}
\begin{align}\label{eq:measurments_a}
\mb{y} &= {\mb{C}}{\mb{x}}-\us, \ \  \mb{C} = C \otimes I_2, \ \ \Mc=\{\Mc_{\rm p}, \Mc_{\rm v} \},
\\\label{eq:measurments_b}
 {C}&=\diag\paren{C_{\rm p}, C_{\rm v}},
 \\
  {C}_{\rm p} &= \col \big( \efrak^{\s \top}_{{\rm p}_1}, \efrak^{\s \top}_{{\rm p}_2}, \dots, \efrak^{\s \top}_{{\rm p}_{\s {|\Mc_{\rm p}|}}} \big) \in \real^{\s |\Mc_{\rm p}| \times N}, 
\\
{C}_{\rm v} &= \col \big(\efrak^{\s \top}_{{\rm v}_1}, \efrak^{\s \top}_{{\rm v}_2}, \dots, \efrak^{\s \top}_{{\rm v}_{\s {|\Mc_{\rm v}|}}} \big) \in \real^{\s |\Mc_{\rm v}| \times N}, 
\end{align}
\end{subequations}
%
where $\us =\col \big( \mb{u}_{{\rm s}_1}, \mb{u}_{{\rm s}_2}, \dots , \mb{u}_{{\rm s}_{\s |\Mc|}} \big) \! \in \! \real^{\s 2|\Mc|}$ denotes the vector-valued sensory attacks on the measurements.

\begin{proposition}\btitle{Formation convergence}\label{prop:formation_convergence}
Assume that the formation configuration is feasible and that the communication graph is connected in each mode. 
Then, under the control protocol \eqref{eq:ctrl_proto}, and in the absence of attacks, the states of the UAVs in \eqref{eq:ol_sys} converge to the desired formation configuration in \eqref{eq:formation_consensus}.
\end{proposition}
\begin{proof}
\ifshort The proof follows a change of variables as in \cite{ren2007consensus} and a convergence analysis similar to that in \cite{mao2020novel}. The details are omitted here due to space limitation.
\else
See Appendix \ref{app_formation_convergence}. \fi
\end{proof}

Note that the UAV's dynamics in \eqref{eq:ol_sys} as well as the control protocol \eqref{eq:ctrl_proto_normal} for the $x$ and $y$ directions are decoupled. Thus, for notational simplicity, we may use the following
\begin{subequations}\label{eq:cl_sys_scalar}
\begin{align}\label{eq:cl_sys_scalar_dyn}
\Sigma:
\dot{\rm{x}}
    &=
    {A}_{\sigma(t)}
    {\rm x}
    +
    {B}^{\boldsymbol{F}}_{\sigma(t)} \rxtar
    + 
  {B}_{\! \Ac} \uar, \ \ {\rm{x}}(t_0)={\rm{x}}_0, 
  \\\label{eq:cl_sys_scalar_measurements}
  {\rm y} &= {{C}}{\rm{x}}-\usr,
\end{align}
\end{subequations}
to represent the dynamics in \eqref{eq:cl_sys} with its monitored states in \eqref{eq:measurments_b} in only one direction of the $x\!-\!y$ plane. Accordingly, $\xr=\col \paren{\pr,\vr} \in \real^{2N}$ in \eqref{eq:cl_sys_scalar} denotes the stacked vector of all positions $\pr = \col \paren{\pr_i}_{i \in \Vc}$ and velocities $\vr = \col \paren{\vr_i}_{i \in \Vc}$ in one direction with their corresponding formation references $\rxtar \in \real^{2N} $ as well as attack inputs $\uar \in \real^{\s |\Ac|} $ and $\usr \in \real^{\s |\Mc|}$, and other system matrices are given in \eqref{eq:A_B_matrices}-\eqref{eq:sys_matrices__attack_sig} and \eqref{eq:measurments_b} with ${B}^{\boldsymbol{F}}_{\sigma(t)} = - {A}_{\sigma(t)}$.
%
\subsection{Attack stealthiness}
Motivated by the susceptibility of wireless networks to adversarial intrusions \cite{manesh2019cyber,yaacoub2020security}, we study the worst-case scenario adversarial settings where an attacker leverages \emph{a priori} system knowledge of the UAVs' coordination or a prerecorded sequence of sensory data to design sophisticated stealthy attacks implementable through actuator attacks ${\ub}_{{\rm a}_i}\!(t)\text{'s},\ i \in \Ac$ in \eqref{eq:ctrl_proto} and sensor attacks $\us(t)$ in \eqref{eq:measurments_a}.

Here, \emph{a priori} system knowledge refers to the initial configuration of the networked system \eqref{eq:cl_sys} with the measurements \eqref{eq:measurments} (or equivalently  \eqref{eq:cl_sys_scalar}), denoted by the tuple $\hat{\Sigma}(\hat{\mb{A}}_{\sigt}, \hat{\mb{C}}, \sigma(t)=1)$ with $\hat{\mb{A}}_{\sigt}$ and $\hat{\mb{C}}$ being the approximations of their counterparts in \eqref{eq:cl_sys} and \eqref{eq:measurments}. The amount of \emph{a priori} system knowledge needed for designing \emph{stealthy} attacks varies for different attacks \cite{teixeira2015secure}, and will be quantified in Section \ref{Sec:stealthy_attacks_realization}.

\emph{Stealthy} attacks refer to a class of adversarial intrusions  (cyber attacks \cite{smith2015covert,mo2009secure,9683427}) ${\ub}_{{\rm a}_i}\text{'s},\ i \in \Ac$ in \eqref{eq:ctrl_proto} and $\us$ in \eqref{eq:measurments_a} that disrupt the system's normal operation while remain stealthy in the monitored measurements \eqref{eq:measurments}, that is 
\begin{equation}\label{eq:stealthy}
   \mb{y}_{}(\mb{x}_0,\ua,\us,t)=\mb{y}^{\rm n}({\mb{x}}^{\rm n}_0,\mb{0},\mb{0},t),\quad \forall \, t\in [t_0,t_{d}), 
\end{equation}
where $\mb{y}^{\rm n}({\mb{x}}^{\rm n}_0,\mb{0},\mb{0},t)=\mb{C}\mb{x}^{\rm n}$ is the output associated with an attack-free system with the same dynamics as in \eqref{eq:cl_sys_scalar_dyn}, and $ \mb{x}_0$ and ${\mb{x}}^{\rm n}_0$ are the actual and a possible initial states, respectively. Also, $t_0$ is the initial time instant and $t_d$ is the attack detection time instant, i.e., the time instant at which condition in \eqref{eq:stealthy} no longer holds and attacks lose their stealthiness.

\subsection{Problem statement: attack detection}
We consider the attack detection problem as a hypothesis testing problem with the null and alternative hypotheses
\begin{align}\label{eq:hypothesis_testing}
    \nullH : {\rm attack}\text{-}{\rm free},\ \text{vs.} \ \
    \alterH: {\rm attacked},
\end{align}
for which we present detection frameworks in Section \ref{Sec:detection_framework}.

\section{Observer Design and Analysis for Attack Detection}
In this section, we characterize the models for stealthy attacks on the networked UAVs in \eqref{eq:cl_sys} and develop centralized and decentralized detection schemes.

\subsection{Realization of stealthy attacks}\label{Sec:stealthy_attacks_realization}
Given system in \eqref{eq:cl_sys}, let $\Mc$ in \eqref{eq:measurments_a} be a set of monitored states and let $\Ac$ be a set of compromised UAVs subject to attack  ${\ub}_{{\rm a}_i} \neq \boldsymbol{0} $ in \eqref{eq:ctrl_proto}. In what follows, we characterize stealthy attacks in terms of different realizations of \eqref{eq:stealthy}.

\noindent
\emph{{Zero-dynamics attack (ZDA).}} 
ZDA refers to the class of attacks based on the zero dynamics of the system $({A}_{\sigt},{B}_{\! \Ac},{C}, \sigt=1)$ in \eqref{eq:cl_sys_scalar} that are (nontrivial) state trajectories excited through input directions $B_{\! \Ac}$ and invisible at the output ${\rm{y}}$, and that can be characterized by the rank deficiency of matrix pencil

\begin{equation}\label{eq:zda_pencil}
P(\lambda_{o}) = 
    \begin{bmatrix}
        \lambda_{o} I_{N} -{A}_{1} & -{B}_{\! \Ac} \\ {C}_{} & \boldsymbol{0}_{}
    \end{bmatrix},
\end{equation}
for some $\lambda_{o} \in \realpos$ \cite{mao2020novel}.

\begin{proposition}\label{prop:ZDA_design}
Assume the system in \eqref{eq:cl_sys} in its initial active mode $\sigt = 1$ has unstable zero dynamics, i.e., the matrix pencil $P(\lambda_{o})$ in 
\eqref{eq:zda_pencil}
is rank deficient for some $\lambda^x_{o},\  \lambda^y_{o}, \in \realpos$, and that the attacker's a priori knowledge of the system   $\hat{\Sigma}(\hat{\mb{A}}_{\sigt}, \hat{\mb{C}}, \sigma(t)=1) = ({\mb{A}}_{\sigt}, {\mb{C}}, \sigma(t)=1) $. Then,
there exists a stealthy attack policy
\begin{align}\label{eq:ZDA_signal}
    \ua \!=\! \col \paren{ {\ub}_{{\rm a}_{i}} }_{i \in \Ac},\  {\ub}_{{\rm a}_{i}} 
    \!=\!
    [{\rm u}^x_{{\rm a}_i}(0)e^{\lambda^x_{o}t}\ \ {\rm u}^y_{{\rm a}_i}(0)e^{\lambda^y_{o}t}\ ]^{\top}\!,
\end{align}
in dynamics \eqref{eq:cl_sys} that causes part of system states exponentially deviate from the formation configuration in \eqref{eq:formation_consensus} while the condition in \eqref{eq:stealthy} holds. In this attack model, the measurement signals are not compromised, i.e., $\us = \boldsymbol{0}$.
\end{proposition}

\begin{proof} The proof follows from \cite[Prop. II.4]{9683427} and so is omitted here.
\end{proof}

It is noteworthy that the assumption on \emph{a priori} system knowledge in Proposition \ref{prop:ZDA_design} can be relaxed. In the cases that only a subset of the system model as \emph{a priori} is disclosed to the attacker, that is $\hat{\Sigma}(\hat{\mb{A}}_{\sigt}, \hat{\mb{C}}, \sigma(t)=1) \approx ({\mb{A}}_{\sigt}, {\mb{C}}, \sigma(t)=1) $, a ZDA can be realized that only affects the UAVs within the known subset of the system, which is known as local ZDA~\cite{teixeira2015secure}.

\noindent
\emph{{Covert attack}}~\cite{smith2015covert}. Covert attacks are a class of intrusions through input channels ${\mb{B}_{\! \Ac}}$ whose covertness at the output is obtained by alteration of the measurement signals \eqref{eq:measurments_a} and whose realization requires perfect knowledge of the system i.e. $\hat{\Sigma}(\hat{\mb{A}}_{\sigt}, \hat{\mb{C}}, \sigma(t)=1) = ({\mb{A}}_{\sigt}, {\mb{C}}, \sigma(t)=1) $. Let attack policy $\ua (t)= \col \paren{ {\ub}_{{\rm a}_{i}} }_{i \in \Ac} : \realnonneg \mapsto \real^2 $ in \eqref{eq:cl_sys} be any continuous signal initiated at time instant $t_{\rm a} \in \realnonneg$. Then, the attack is covert and \eqref{eq:stealthy} holds if the attacker alters the measurement \eqref{eq:measurments_a} by
\begin{align}\label{eq:covert_us}
\us(t)=
    {\mb{C}}
	{\int^{t}_{t_{\rm a}}
	e^{{\mb{A}}_1(t-{\tau})}
	\mb{B}_{\! \Ac}\ua({\tau})d{{\tau}}}.
\end{align}
We refer to \cite{9683427} for details of the derivation and proof.

\noindent
\emph{{Cooperative DoS and replay attack.}} It is shown that a denial-of-service (DoS), interfering in a UAVs' communication, causes unstable and unsafe flights \cite{chen2019container}. We formulate a scenario where replay attacks\footnote{A replay attack is the case that the attacker replays (periodically) a sequence of stored data as real-time measurements to conceal any deviation from a normal operation.} \cite{mo2009secure} are implemented in cooperation with a DoS in order to keep the DoS stealthy in the networked-level measurements \eqref{eq:measurments_a}. Here, the cooperatively-stealthy DoS and replay attack take place when the UAVs have reached the formation configuration in \eqref{eq:formation_consensus} and thus are hovering only, giving the attacker the opportunity to record and store slow-varying measurements \eqref{eq:measurments_a} for a time interval $T_r \in \realpos$ before starting the attacks $\ua $ and $ \us$ respectively in \eqref{eq:cl_sys} and  \eqref{eq:measurments_a}  that is $\ua(t) = \boldsymbol{0} $ and  $\us(t) = \boldsymbol{0}, \;  \forall \, t \in [t_0\;  t_{\rm a}) $ where $t_{\rm a} > T_r$.
Then, upon starting a DoS at a time instant $ t_{\rm a} \in \realpos$, causing one or more UAVs lose their inter-communication and deviating from the equilibrium states \eqref{eq:formation_consensus}, a concurrent replay attack $\us$ in $\eqref{eq:measurments_a}$ given by
\begin{align}\label{eq:repay_us}
    \us(t)=\mb{C}\mb{x}(t)-\mb{y}(t-n T_r),\ n \in \naturals, \ \ t \geq t_{\rm a},
\end{align}
causes the stealthiness condition in \eqref{eq:stealthy} holds.

We note that \emph{a priori} system knowledge is not required for the cooperative DoS and replay attack that is $\hat{\Sigma}(\hat{\mb{A}}_{\sigt}, \hat{\mb{C}}, \sigma(t)) = \emptyset $.

\subsection{Observer-based detection framework}\label{Sec:detection_framework}
The susceptibility of observers/monitors to the sophisticated attacks satisfying the stealthiness condition in \eqref{eq:stealthy} has been demonstrated in both deterministic and stochastic settings in \cite{mao2020novel,smith2015covert,mo2009secure}.
Motivated by this challenge, we present centralized and decentralized detection schemes to address the attack detection problem formulated as in \eqref{eq:hypothesis_testing}.

\noindent
\subsubsection{Centralized detection scheme} 
In the centralized detection scheme, we leverage switching links in the inter-UAVs' communication topology to generate model discrepancy rendering the stealthy attacks detectable in the measurements \eqref{eq:measurments} monitored in a ground control center. Note that the UAVs' communication may be subject to switching connections in two senses. First, a communication link failure induced due to operation in uncertain environments, and second, a planned switch (addition or removal of connections) triggered for security and performance reasons. Regardless of the underlying causes of switching links in the inter-UAVs' communication, we investigate their effect on the detection of stealthy attacks.

Consider the dynamical system in \eqref{eq:cl_sys_scalar}, a centralized attack detection monitor (central monitor), derived based on the initial (normal) communication mode of UAVs ($\sigt=1$), is given by

\vspace{1ex}
\noindent
\small
\begin{empheq}[left=\Sigma_{\s \Oc}^{\s \Mc} \!:\!\empheqlbrace]{align}\label{eq:obs_cent}
\dot{\hat{{\rm{x}}}} 
&=
{A}_{\sigt} \hat{\rm x} +{{B}}^{\boldsymbol{F}}_{\sigma(t)} \rxtar + {{H}}_{}({\rm{y}}_{\s} - \hat{\rm y}_{} ),
&  \sigt=1,
\nonumber \\ 
\hat{\rm{y}}&={C}_{\s }\hat{\rm x}, & \hat{\rm{x}}(0)=\boldsymbol{0},
 \\ \nonumber
{\rm{r}}_{ 0} &= {\rm{y}} -\hat{\rm{y}}, & \hspace{-6ex} \text{central residual,} 
\end{empheq}
\normalsize
where $ H $ is an observer gain such that $(A_{\sigt}-HC)$ is stable in all modes and $\lim_{t \rightarrow \infty} {\rm{r}}_{ 0} = \boldsymbol{0}$ in the absence of attacks \cite{9683427}. Also, we let ${\rm r}_{0}^j(t)  = {C}^{j} ({\rm{x}} - \hat{\rm{x}})$ denote the $j$-th component of the residual ${\rm{r}}_{ 0}$ 
with ${C}^{j}$ being the $j$-th row vector of matrix ${C}$ in \eqref{eq:measurments_b}. Then, in the absence of attacks, an upper bound on the residuals is obtained as follows:
\begin{align}\label{eq:cent_threshold}
   | {\rm{r}}_{ 0}^j(t) | \leq  {\bar{k}_j} e^{-\bar\lambda_j t} \bar{\omega} + \epsilon_0 =: \epsilon^j_0,
\end{align}
where ${\bar{k}_j}$ and $\bar\lambda_j$ are positive constants such that $|{C}^{j}e^{({A}_1-{H}_{}{C})t}| \! \leq \! {\bar{k}_j} e^{-\bar\lambda_j t} $, $\bar{\omega}$ is an upper bound such that $|\rm{e}(0)| \! = \! |\rm{x}(0)-\hat{\rm{x}}(0)| \!=\! |\rm{x}(0)| \leq \bar{\omega}$, and  $\epsilon_0 \!\in \! \realpos$ is a sufficiently small constant to account for measurement noises.
 
Given the central monitor \eqref{eq:obs_cent} and its corresponding thresholds in \eqref{eq:cent_threshold}, the hypothesis testing problem in \eqref{eq:hypothesis_testing} can be quantified either by 
\begin{subequations}\label{eq:hypotheses_global}
\begin{align}\label{eq:hypotheses_global_null}
    \nullH &: {\rm attack}\text{-}{\rm free}, 
    & \text{if} & & | {\rm{r}}_{ 0}^j(t) | &\leq  \epsilon^j_0, \ \ \forall \, j \in \Mc,
    \\ \label{eq:hypotheses_global_alternative}
    \alterH &: {\rm attacked}, 
    & \text{if} & & | {\rm{r}}_{ 0}^j(t) | &>  \epsilon^j_0, \ \ \exists \, j \in \Mc,
\end{align}
\end{subequations}
or by
\vspace{-1em}
\begin{align}
    {\rm{r}}_{ 0}^{\top}\boldsymbol{\Sigma}^{-1}_{{\rm{r}}_{\s 0}}{\rm{r}}_{ 0}
    \overset{\nullH}{\underset{\alterH}{\lesseqgtr}} 
    \texttt{threshold},
    \vspace{-1ex}
\end{align}
with $\boldsymbol{\Sigma}_{{\rm{r}}_{\s 0}} $ being the covariance of the residual ${\rm{r}}_{ 0}$ having a zero-mean Gaussian distribution in stochastic settings where a discretized version of \eqref{eq:obs_cent} as a Kalman filter is used, together with $\chi^2$ (chi-squared) tests, for attack detection \cite{mo2009secure}.

As shown in \cite{mo2009secure,9683427}, $\chi^2$ detectors, Kalman filters, and Luenberger-type observers/monitors fail in detecting the stealthy attacks that were presented in Section \ref{Sec:stealthy_attacks_realization} provided the stealthiness condition \eqref{eq:stealthy} holds, causing a false validation of the null hypothesis \eqref{eq:hypotheses_global_null}. 
Here, based on the results in \cite[Th. III.3]{9683427}, we evaluate the effect of switching connections in inter-UAVs’
communication on the violation of \eqref{eq:stealthy} and thus on the validation of the null hypothesis \eqref{eq:hypotheses_global_null}. This procedure will be presented in Algorithm \ref{alg:detection_alg_global} in Section \ref{Sec:Results}.

\noindent
\subsubsection{Decentralized detection scheme} In the decentralized detection scheme, a set of UAVs, equipped with on-board (local) monitors, leverage the information exchange with their neighbouring UAVs to locally detect the stealthy attacks on their neighbours. Upon attack detection, a local monitor triggers an inter-UAV communication switch and informs other local monitors as part of a contingency plan (see Algorithm \ref{alg:detection_alg_global}).

Note that in the networked UAVs with a connected communication graph $\Gc_{\sigt}$, any UAV has access to the states of itself as well as the position states of the set of its immediate neighbors $\Nc^{i}_{\sigt}$ (cf. control protocol \eqref{eq:ctrl_proto_normal}). Accordingly, we define, for the $i$-th UAV in the network, a set of local measurements, indexed by set  $\Mc^{i} $, as follows:
\begin{subequations}\label{eq:measurments_local}
\begin{align}\label{eq:measurments_local_a}
    \Mc^{i} &= \Nc^{i}_{\sigt} \cup \{i\}, & \hspace{-2em} \sigt &=1 \in \Qc,
    \\ \label{eq:measurments_local_b}
    {\mb{y}}_{i} &= \mb{C}_{i}{\mb x},\ \,  \text{and} \ \, {\rm{y}}_{i} = {C}_{i}{\rm x},
    & \mb{C}_{i} &= C_i \otimes I_2,
    \\ \label{eq:measurments_local_b}
    C_i &= \diag\paren{{\rm C}_{{\rm p},i}, \efrak_{i}^{\s \top}},
    &
    {\rm C}_{{\rm p},i} &= \col \big(\efrak^{\s \top}_{j} \big)_{j \in \Mc^{i}}. 
\end{align}
\end{subequations}
where $\mb{x}$ and ${\rm x}$ are the system states in \eqref{eq:cl_sys} and \eqref{eq:cl_sys_scalar}, respectively. Different from the networked measurements in \eqref{eq:measurments}, the local measurements in \eqref{eq:measurments_local} are not transmitted through compromised network channels to the control center for monitoring. Instead, they are locally available for each UAV, and thus are not subject to alterations by sensory attacks. 

Given the local measurements \eqref{eq:measurments_local} and dynamics \eqref{eq:cl_sys_scalar}, we define the local attack detector $\Sigma^{\, i}_{\s \Oc}$ for the $i$-th UAV as follows:
\begin{equation}\label{eq:obs_decent}
 \Sigma^{\, i}_{\s \Oc}\!:\!
 \left\{
 \begin{array}{ll}
\hspace{-1ex} \dot{\hat{{\rm{x}}}}_i  =
  {A}_{\sigt} \hat{\rm x}_i +{{B}}^{\boldsymbol{F}}_{\sigma(t)} \rxtar + {{H}}^i_{}({\rm{y}}_{i} - \hat{\rm y}_{i} ),
& \sigt \in \Qc,
 \\ 
\hspace{-1ex} \hat{\rm{y}}_i={C}_{i}\hat{\rm x},  &\hat{\rm{x}}(0)=\boldsymbol{0},
 \\
\hspace{-1ex}\, {\rm{r}}_{ i} = {\rm{y}}_i -\hat{\rm{y}}_i, &    \hspace{-4ex} \text{local residual}, 
\end{array}
\right.  
\end{equation}
where $\hat{\rm x}_i$ is the local estimation of ${\rm x}$ in \eqref{eq:cl_sys_scalar}, and $ H^i $ is an observer gain such that $(A_{\sigt}-H^iC_i)$ is stable in all modes. Therefore, in the absence of attacks, $\lim_{t \rightarrow \infty} {\rm{r}}_{ i} = \boldsymbol{0}$, and similar to the central monitor's, the $j$-th component of local residuals, ${\rm{r}}^j_{ i}$'s, hold an upper bound (threshold) as follows:
\begin{align}\label{eq:local_threshold}
   | {\rm{r}}_{i}^j(t) | \leq  {\bar{k}_{i,j}} e^{-\bar\lambda_{i,j} t} \bar{\omega} + \epsilon_i =: \epsilon^j_i,
\end{align}
where ${\bar{k}_{i,j}}$ and $\bar\lambda_{i,j}$ are positive constants such that $|{C}^{j}_{i}e^{({A}_1-{H}^{i}{C^j_i})t}|\leq {\bar{k}_{i,j}} e^{-\bar\lambda_{i,j} t} $, $\bar{\omega}$ is an upper bound such that $|{\rm{e}}_i(0)|= |\rm{x}(0)-\hat{{\rm x}}_i(0)| = |\rm{x}(0)| \leq \bar{\omega}$, and  $\epsilon_i \in \realpos$ is a sufficiently small constant to account for measurement noises. 

Now, the hypothesis testing problem in \eqref{eq:hypothesis_testing} can be revisited and quantified using local residuals as follows: 
\begin{subequations}\label{eq:hypotheses_local}
\begin{align}\label{eq:hypotheses_local_null}
    \nullH &\!:\! {\rm attack}\text{-}{\rm free}, 
    \hspace{-1em}\;
    &  \text{if} \; & | {\rm{r}}_{i}^j(t) | \leq  \epsilon^j_i, \; \forall \, j \in \Mc^i, \, \forall\,i \in \Dc,
    \\\label{eq:hypotheses_local_alternative}
    \alterH &\!:\! {\rm attacked}, 
    & \text{if} \; & | {\rm{r}}_{i}^j(t) | >  \epsilon^j_i, \; \exists \, j \in \Mc^i, \, \exists \, i \in \Dc,
\end{align}
\end{subequations}
where $\Dc$ is the set of all the UAVs equipped with a local detector as in \eqref{eq:obs_decent}.

Note that a successful attack detection using the hypothesis testing \eqref{eq:hypotheses_local} does depend on the sensitivity of the local residuals, ${\rm r}_i$'s, to the stealthy attacks. In this regard, the following results characterize the capability of local detectors in detecting stealthy attacks.
\begin{proposition}\label{prop:detectability_local}
Consider dynamics \eqref{eq:cl_sys_scalar} and let the $i$-th UAV be equipped with the local attack detector $\Sigma^{\, i}_{\s \Oc}$ in \eqref{eq:obs_decent} and local measurements \eqref{eq:measurments_local}. Then, stealthy ZDA and covert attacks are detectable in $\Sigma^{\, i}_{\s \Oc}$'s residual ${\rm r}_i$ if the set of compromised UAVs satisfies $\Ac \subseteq \Nc_{\sigt}^i,\ \sigt=1 \in \Qc$.
\end{proposition}
\begin{proof}
\ifshort Omitted due to space limitation. \else
See Appendix \ref{app_detectability_local}. \fi
\end{proof}

Note that the $i$-th UAV's local monitor, $\Sigma^{\, i}_{\s \Oc},\ i \in \Dc $ secures the networked UAVs against the stealthy attacks on its neighbors' set ${\Nc}^i_{\sigt}$. Therefore, the problem of interest is to determine a set $\Dc \subseteq \Vc $ of local detectors $\Sigma^{\, i}_{\s \Oc}\text{'s},\, i \in \Dc$ such that they cover the entire set $\Vc$ of UAVs.

\begin{proposition}\label{prop:detectability_local_total}
Consider the networked UAVs with the dynamics in  \eqref{eq:cl_sys_scalar} subject to stealthy attacks on a set of compromised UAVs $\Ac \subseteq \Vc$ and let the set
\begin{align}\label{eq:local_obs_set}
   \Dc \coloneqq \{i \in \Vc \mid \bigcup_{i \in \Dc} {\Nc}^i_{\sigt} =\Vc ,\ \sigt= 1 \in \Qc \}, 
\end{align}
represent the set of UAVs equipped with local attack detectors $\Sigma^{\, i}_{\s \Oc}$'s in \eqref{eq:obs_decent}. Then, stealthy ZDA and covert attacks undetectable in $\Sigma^{\, i}_{\s \Oc}$'s residual ${\rm r}_i,\ \forall \, i \in \Dc$, is impossible, securing the entire network set $\Vc$ of UAVs against stealthy attacks.
\end{proposition} 

\begin{proof}
\ifshort Omitted due to space limitation. \else
See Appendix \ref{app_detectability_local_total}. \fi
\end{proof}
\begin{figure}[t]
  \subfloat[Formation References\label{fig:formation_graph}]{\small
    \includegraphics[width=.35\linewidth]{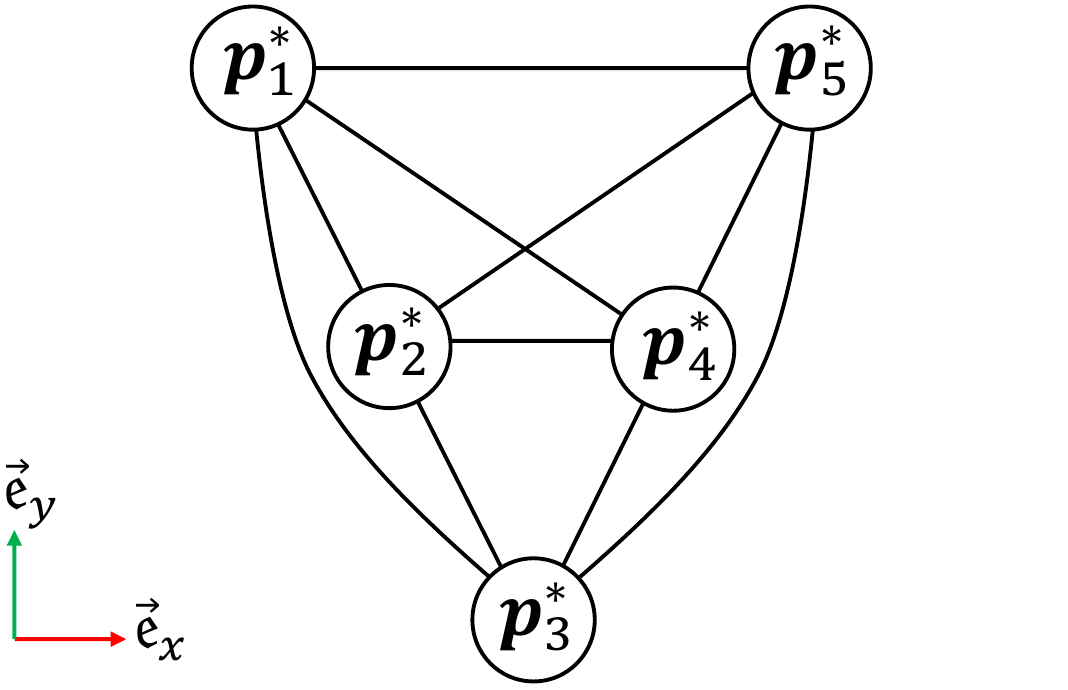}
  }
 \hfill
 \subfloat[V-shape Formation\label{fig:V_shapeformation_graph}]{
    \includegraphics[width=.54\linewidth]{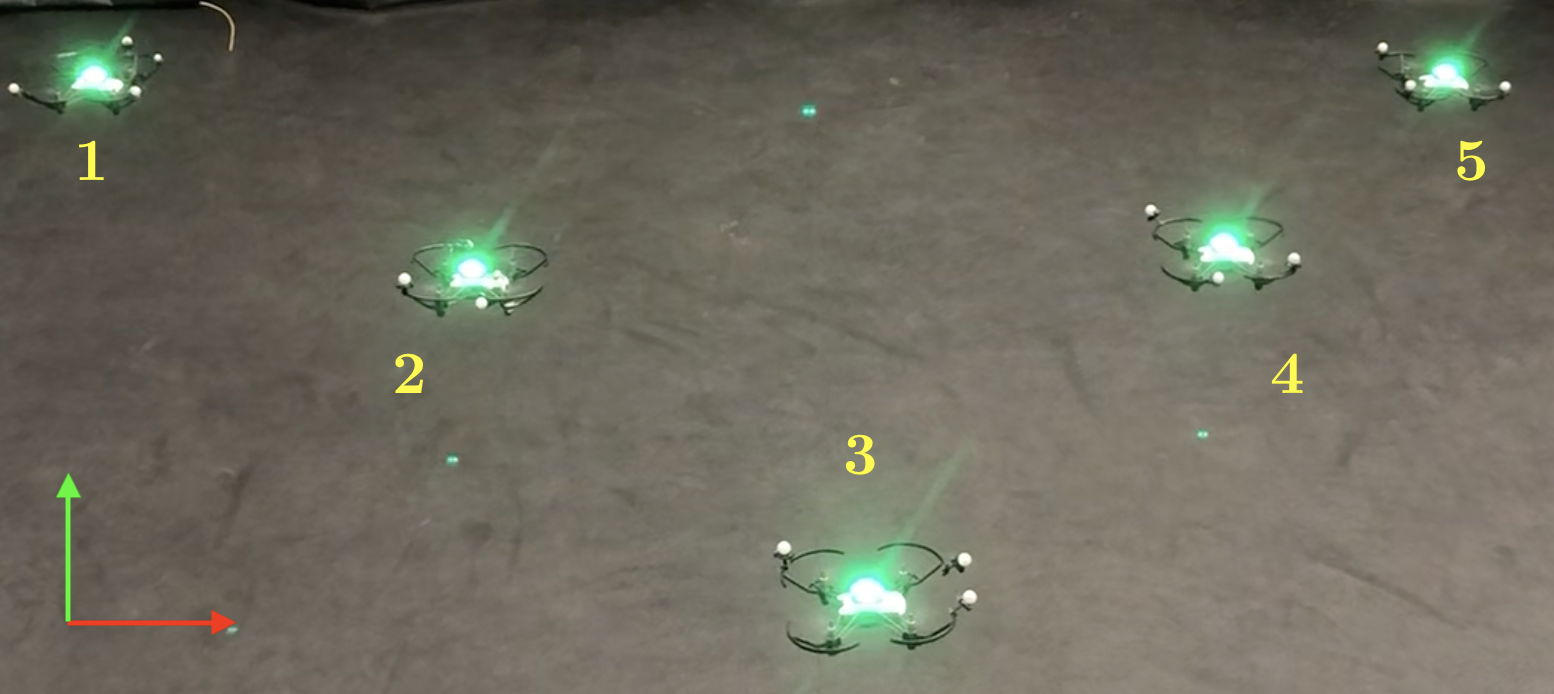}
  }
\\
  \subfloat[mode 1\label{fig:graph1}]{\small
    \includegraphics[width=.23\linewidth]{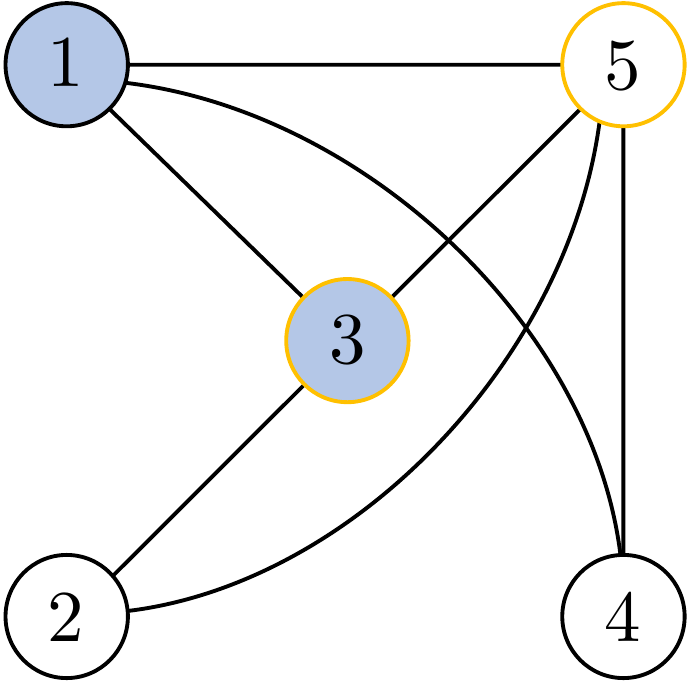}
  }
  \subfloat[mode 2\label{fig:graph2}]{\small
    \includegraphics[width=0.23\linewidth]{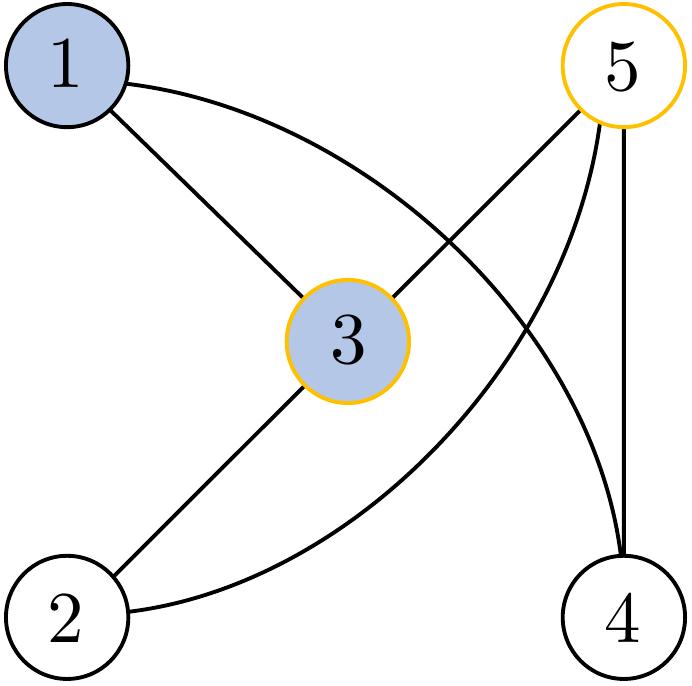}
  }
  \subfloat[mode 3\label{fig:graph3}]{\small
    \includegraphics[width=0.23\linewidth]{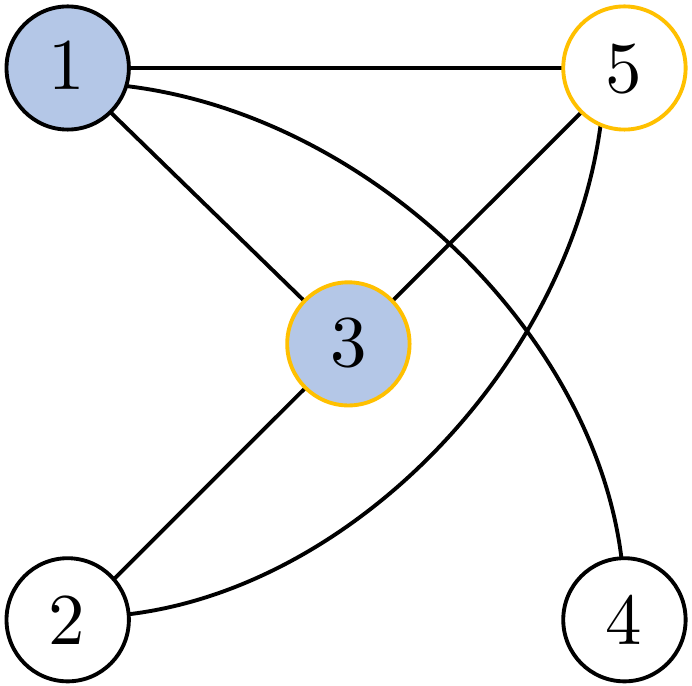}
  }
  \subfloat[mode 4\label{fig:graph4}]{\small
  \includegraphics[width=0.23\linewidth]{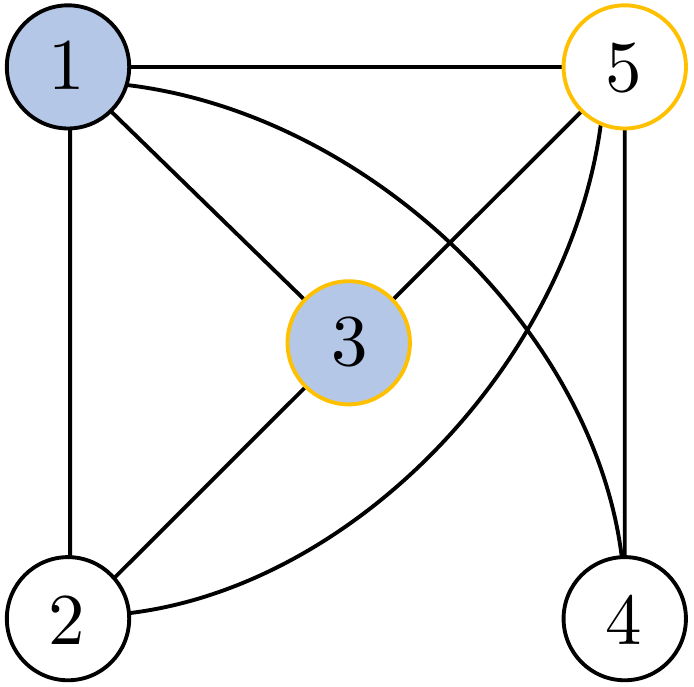}}
  %
  \caption{ Multi-UAV's formation and communication topology. (a) Formation references specifying a V-shape in the $x\!-\!y$ plane. (b) V-shape formation of UAVs. (c)-(f) Inter-UAV's communication graph $\Gc_{\sigt}$ with four modes $ \sigt = \{1,2,3,4\}=: \Qc $. UAVs initially communicate in mode $\sigt=1$ and may switch to other modes $\sigt=\{2,3,4\}$ if activated by a local detector. Blue nodes indicate the UAVs equipped with a local monitor and orange nodes specify the UAVs monitored by the ground control center.}
  \label{fig:graphs}
\vspace{-1em}
\end{figure}
\vspace*{-1em}
\begin{algorithm}[H]
\small
\renewcommand{\algorithmicrequire}{\textbf{Input:}}
\renewcommand{\algorithmicensure}{\textbf{Require:}}
\def\NoNumber#1{{\def\alglinenumber##1{}\State#1}\addtocounter{ALG@line}{-1}}
\caption{\small Attack detection by the $i$-th local monitor, $i \in \Dc$}\label{alg:detection_alg_local}
\begin{algorithmic}[1]
\Require $\Sigma^{\, i}_{\s \Oc},\ i \in \Dc $ in \eqref{eq:obs_decent} and \eqref{eq:local_obs_set}, $\mb{y}_{i}$ in \eqref{eq:measurments_local}, $ \epsilon^j_i$ in  \eqref{eq:local_threshold}
\Procedure{local hypothesis testing \eqref{eq:hypotheses_local}}{}
  	\While {$\nullH$ in \eqref{eq:hypotheses_local_null}} 
  	\State Compute local residual ${\rm r}_i$ as in \eqref{eq:obs_decent} 
  	\State Compute corresponding thresholds $\epsilon^j_i$ as in \eqref{eq:local_threshold}
    \If{$ |{\rm r}^j_i|  > \epsilon^j_i $}
        \State Reject the null hypothesis $\nullH$ in \eqref{eq:hypotheses_local_null} \Comment{{\color{gray}Stealthy}}
        \NoNumber{ {\color{gray} attack is locally detected.}}
        \State cooperate with other local detectors in $\Dc$ to
        \NoNumber{run a contingency plan for the entire network.}
    \EndIf
    \EndWhile
\EndProcedure
\end{algorithmic}
\normalsize 
\end{algorithm}
\vspace*{-1em}
\begin{algorithm}[H]
\small
\caption{\small Topology switching for centralized attack detection}\label{alg:detection_alg_global}
\renewcommand{\algorithmicrequire}{\textbf{Inputs:}}
\renewcommand{\algorithmicensure}{\textbf{Require:}}
\def\NoNumber#1{{\def\alglinenumber##1{}\State#1}\addtocounter{ALG@line}{-1}}
\begin{algorithmic}[1]
\Require local observer: $\Sigma^{\, i}_{\s \Oc},\ i \in \Dc $ in \eqref{eq:obs_decent} and \eqref{eq:local_obs_set}, $\mb{y}_{i}$ in \eqref{eq:measurments_local}, $ \epsilon^j_i$ in  \eqref{eq:local_threshold};

centralized observer: $\Sigma_{\s \Oc}^{\s \Mc}$ in \eqref{eq:obs_cent}, $\mb{y}_{}$ in \eqref{eq:measurments}, $ \epsilon^j_0$ in  \eqref{eq:cent_threshold}
\Procedure{Central hypothesis testing \eqref{eq:hypotheses_global}}{} 
\State Run Algorithm \ref{alg:detection_alg_local}
    \If{$\alterH$ in \eqref{eq:hypotheses_local_alternative}}
        \State Switch to a new comm. mode $\sigt \in \Qc \setminus{1}$ 
        \Comment{{\color{gray} Stealthy }}
        \NoNumber{ {\color{gray} attack has been detected locally.}}
    \EndIf
\While {$\nullH$ in \eqref{eq:hypotheses_global_null}} 
     \State Compute central residual ${\rm r}_0$ as in \eqref{eq:obs_cent} 
  	\State Compute corresponding thresholds $\epsilon^j_0$ as in \eqref{eq:cent_threshold}
    \If{$ |{\rm r}^j_0|  > \epsilon^j_0 $}
        \State Reject the null hypothesis $\nullH$ in \eqref{eq:hypotheses_global_alternative}
        \Comment{{\color{gray}Stealthy}}
        \NoNumber{ {\color{gray}attack is detected globally at the control center.}}
    \EndIf
\EndWhile
\EndProcedure
\end{algorithmic}
\normalsize
\end{algorithm}
%
%
%
%
\begin{figure*}[ht]
  \subfloat[UAVs' position in the $x\!-\!y$ plane.\label{fig:ZDA_neutral_sw_position_XY_plane}]{\small
    \includegraphics[width=.49\linewidth]{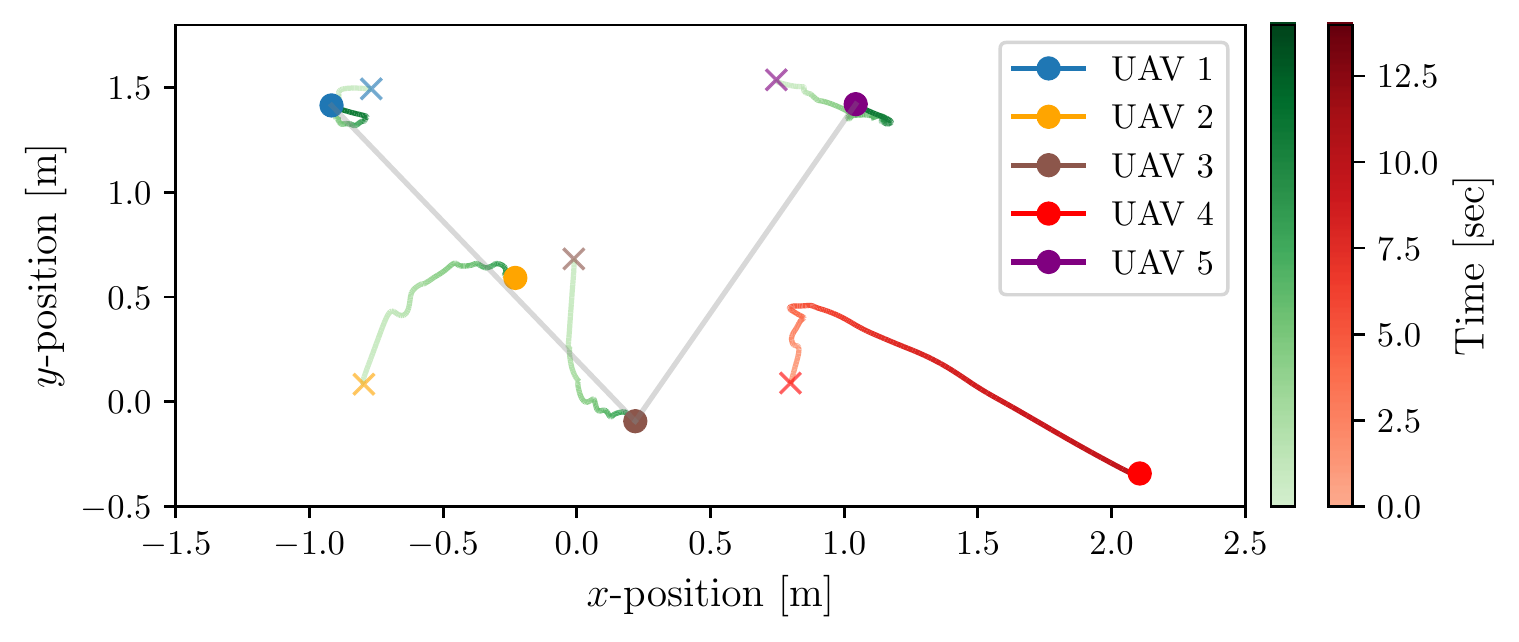}
  }
 \subfloat[Coordination of UAVs in the $y$ direction. \label{fig:ZDA_neutral_sw_relative_disp_y_axis}]{\small
    \includegraphics[width=.49\linewidth]{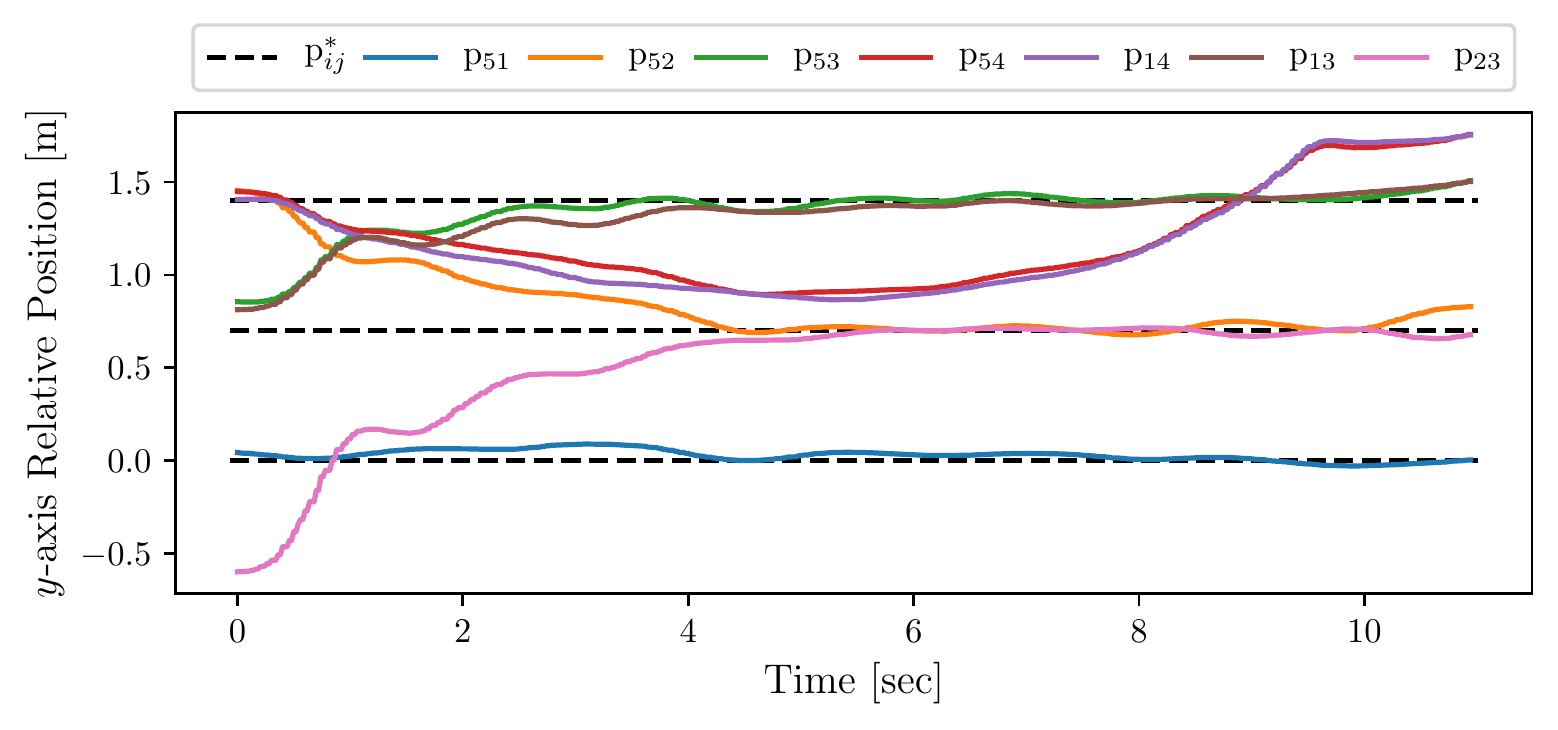}
  }
  \caption{Experiment 1: ZDA on UAVs $1,4,5$ and topology switching from mode $1$ to $4$. (a) UAVs' position in the $x\!-\!y$ plane with the colorbars quantifying the time span. The $\times$ markers and the colored circles show, respectively, the UAVs' initial position and final position during the experiment. Finally, the gray lines visualize the V-shape formation achieved by the final position of the UAVs. (b) The relative positions of UAVs in the $y$ direction corresponding to the inter-UAV communication links in mode $\sigt = 1$, shown in Figure \ref{fig:graph1}. Also, the dashed lines, labeled by ${\rm p}^{*}_{ij},\ i,j \in \Vc$, denote the desired relative positions based on the formation references in Figure \ref{fig:formation_graph}.}
  \label{fig:exp_1_zda}
\vspace{-1ex}
\end{figure*}

It is worth mentioning that a trivial solution for \eqref{eq:local_obs_set} is $\Dc=\Vc$ that is all of the UAVs are equipped with a local detector, although this set can be optimally selected.
\begin{figure}[h]
  \subfloat[Residuals of local monitor $\Sigma^{\, 1}_{\s \Oc}$ run on UAV $1$. The stealthy ZDA is locally detected at $t = 3.22\ \si{sec} $.\label{fig:ZDA_neutral_sw_local_residual_1_XY_small}]{\small
    \includegraphics[width=.98\linewidth]{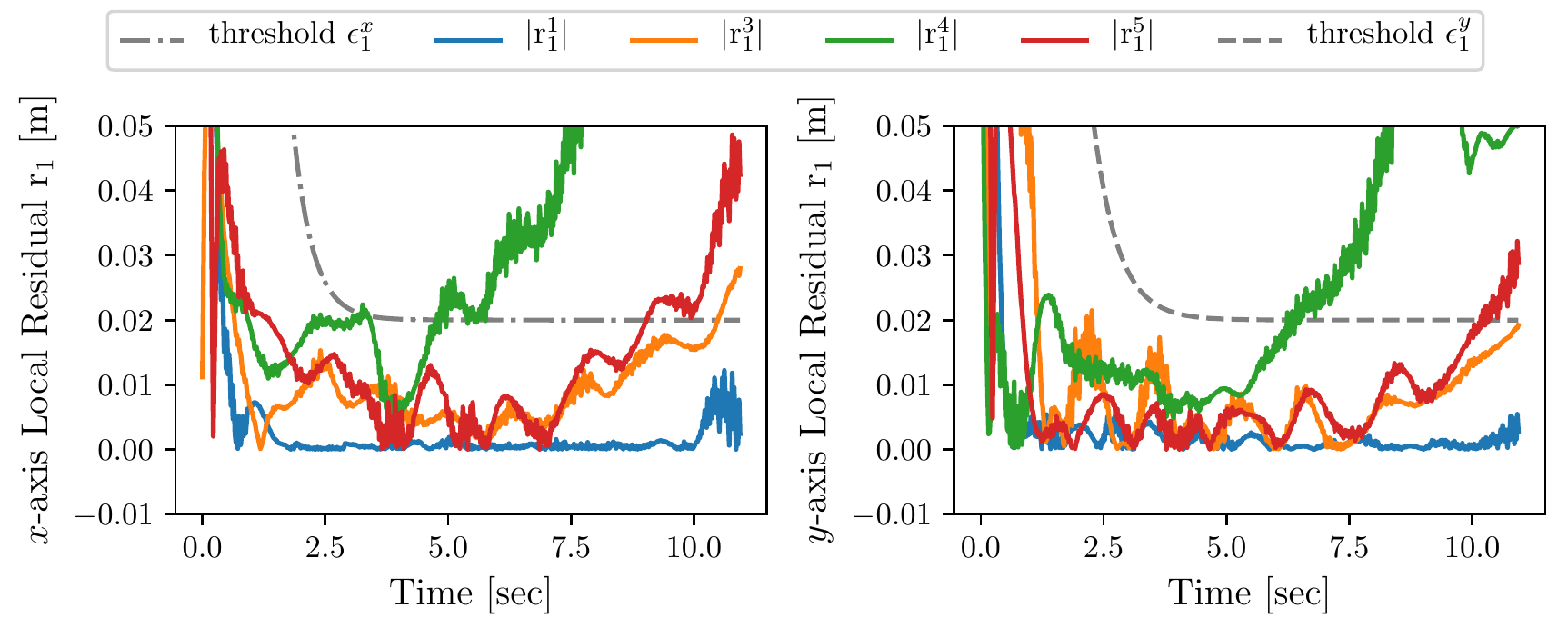}
  }
  \\
 \subfloat[Residuals of local monitor $\Sigma^{\, 3}_{\s \Oc}$ run on UAV $3$. \label{fig:ZDA_neutral_sw_local_residual_3_XY_small}]{\small
    \includegraphics[width=.98\linewidth]{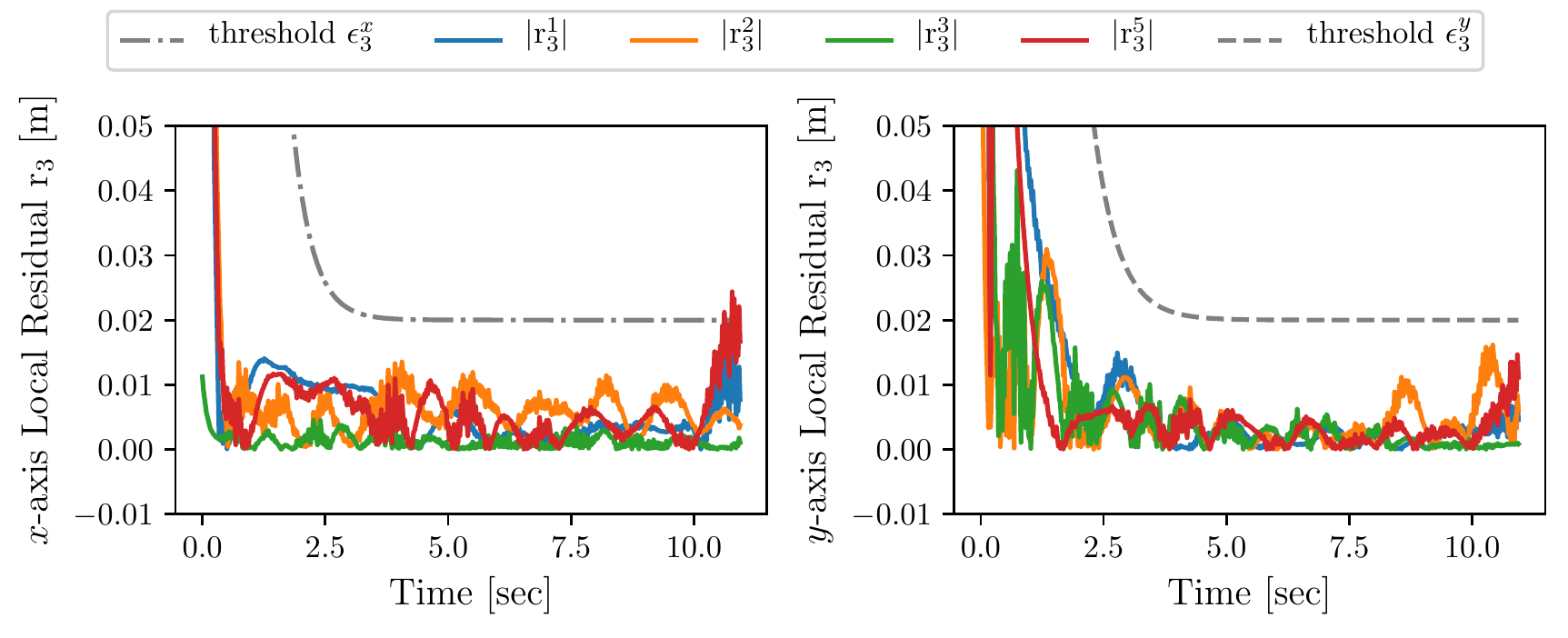}
  }
  \\ 
  \subfloat[Residuals of central monitor $\Sigma^{\s \Mc}_{\s \Oc}$ run on the control center. \label{fig:ZDA_neutral_sw_central_residual_XY_small}]{\small
    \includegraphics[width=.98\linewidth]{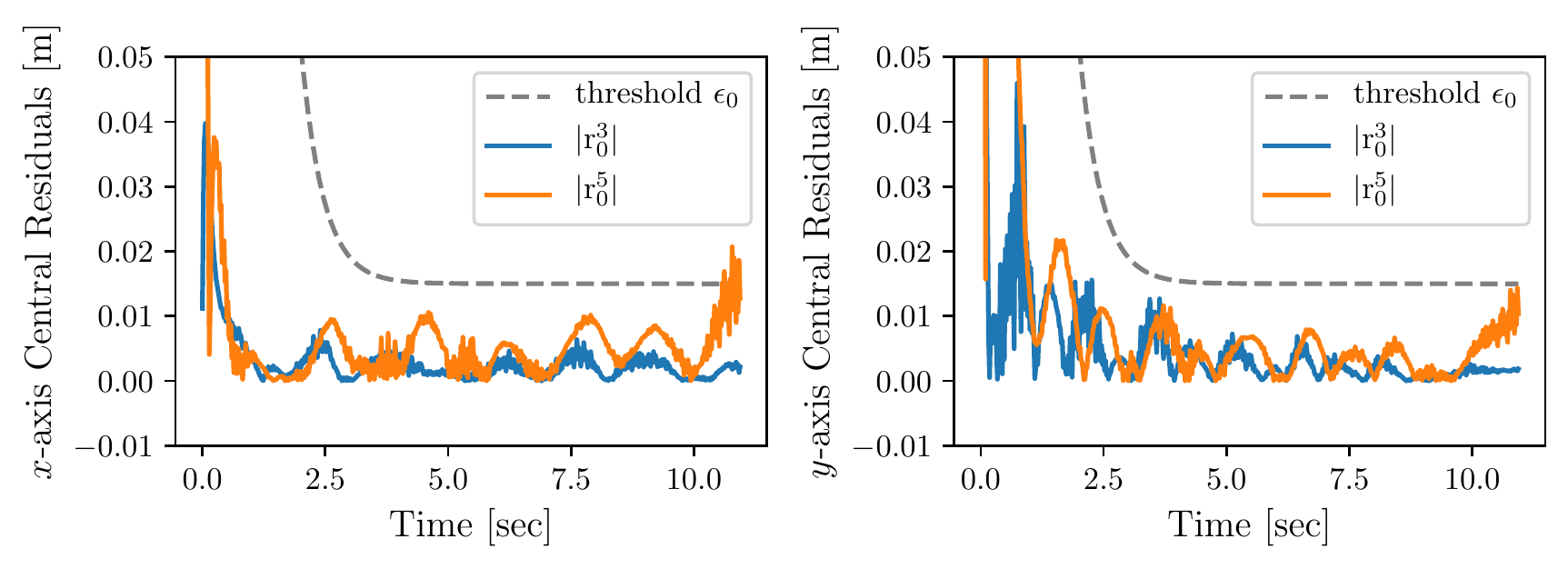}
  } 
  \caption{Experiment 1: ZDA on UAVs $1,4,5$ and topology switching from mode $1$ to $4$, which is triggered by local monitor $\Sigma^{\, 1}_{\s \Oc}$ at $t = 3.22\ \si{sec} $. (a)-(b) The notation $ {\rm r}^{i}_{1} ,\ i \in \{1,3,4,5\}$ ($ {\rm r}^{i}_{3} ,\ i \in \{1,2,3,5\}$), denotes the residual of position estimation for the UAV $1$'s ($3$'s) neighbors obtained by its local monitor $\Sigma^{\, 1}_{\s \Oc}$ ($\Sigma^{\, 3}_{\s \Oc}$) in the $x$ and $y$ directions with the respective thresholds $\epsilon^{x}_1$ ($\epsilon^{x}_3$) and $\epsilon^{y}_1$ ($\epsilon^{y}_3$) as given in \eqref{eq:local_threshold}. (c) The notation $ {\rm r}^{i}_{0},\ i \in \{3,5\}$, denotes  the residual of position estimation for UAVs $3$ and $5$ by the central monitor $\Sigma^{\s \Mc}_{\s \Oc}$ in the $x$ and $y$ directions with the threshold $\epsilon_0$ as given in \eqref{eq:cent_threshold}.}
  \label{fig:exp_1_zda_residuals}
\vspace{-1ex}
\end{figure}

Given Propositions \ref{prop:detectability_local} and \ref{prop:detectability_local_total}, one can verify that the networked UAVs can be secured against stealthy attacks using a set of local monitors, given by \eqref{eq:obs_decent}, that locally detect stealthy attacks, addressing problem \eqref{eq:hypotheses_local}. A procedure for this local hypothesis testing will be presented in Algorithm \ref{alg:detection_alg_local} in Section \ref{Sec:Results}. 

\section{Experimental Results}\label{Sec:Results}
We conducted a set of experiments that serve two purposes. First, the evaluation of the stealthiness of intrusions/deception attacks, described in Section \ref{Sec:stealthy_attacks_realization}, on the wireless communication network of a team of quadrotor UAVs in real-time practical settings. Second, the performance evaluation of the detection schemes presented in Section \ref{Sec:detection_framework}. 
\subsection{Experimental Setup}
Our experimental setup consists of a team of five homogeneous quadrotors (Tello Drones\footnote{https://www.ryzerobotics.com/tello.}), shown in Figure \ref{fig:tello}, flying in a $ 6\ \si{m} \times 4\ \si{m} \times 3\ \si{m} $ flight area that is equipped with  a VICON\footnote{https://www.vicon.com.} motion capture system with 10 cameras. The VICON system provides the ground truth position and orientation of each UAV at $50\ \si{Hz}$ for a central PC running Ubuntu 20.04 with ROS Noetic.  


In our experiments, we use the VICON system's ground truth data available in the central PC to compute the high-level formation control commands that are sent to each UAV at $50\ \si{Hz}$ and to run the central and local monitors, presented in Section \ref{Sec:detection_framework}. The central PC transmits the high-level formation control commands to the UAVs through different Wi-Fi channels and the UAVs' on-board attitude controllers use the received control commands to stabilize and steer the UAVs to their desired pose (see Figure \ref{fig:quad_coord}b). 
This connection setup allows us to replicate the peer-to-peer communication of the UAVs and also to implement stealthily intrusions on the Wi-Fi channels in a controlled setting.
\subsection{Results}
We performed several experiments that serve as a proof of concept of the real-world applicability of the proposed attack detection methods in multi-UAV cooperation settings. In these tests, the UAVs are tasked with achieving a V-shape formation. The spacial configuration of the V-shape formation and a picture of its real-world implementation are shown in Figures \ref{fig:formation_graph} and \ref{fig:V_shapeformation_graph}, respectively. 

In our experiments, the UAVs, indexed by $\Vc = \{1,2,3,4,5\}$, coordinate using the control protocol \eqref{eq:ctrl_proto_normal}, initially in communication mode $\sigt=1$, shown in Figure \ref{fig:graph1}, to achieve the V-shape formation. The UAVs also have consensus on their yaw angle $\psi_i = \psi^* = 0,\ \forall \, i \in \Vc $ as well as their hovering altitude.
We select the position of UAVs $3$ and $5$ as the network-level monitored states at the ground control center that is $\Mc_{\rm p} = \{3,5\}$ and $\Mc_{\rm v} = \emptyset $ in \eqref{eq:measurments}. These measurements are used in the realization of the central monitor (attack detector $\Sigma^{\s \Mc}_{\s \Oc} $) in \eqref{eq:obs_cent} and its residuals in \eqref{eq:cent_threshold}. 
We also let $\Dc = \{ 1, 3 \}$ in \eqref{eq:local_obs_set}, that is UAVs $1$ and $3$, which have the respective set of neighbors $\Nc^{1}_{\sigt =1} = \{3,4,5\}$ and $\Nc^{3}_{\sigt =1}  = \{1,2,5\}$, and the local measurements $\mb{y}_{i} (\text{or}\, {\rm y}_{i}),\ i \in \Dc$ in \eqref{eq:measurments_local}, are selected as the host UAVs for local monitors (attack detectors $\Sigma^{\, i}_{\s \Oc} $'s) in \eqref{eq:obs_decent}.
Accordingly, the condition \eqref{eq:local_obs_set} holds which in turn guarantees the local monitors of UAVs $1$ and $3$ are sufficient to locally detect the stealthy attacks on the entire network of UAVs in a decentralized manner. 
Also, as described earlier in Sections \ref{Sec:formation_control}, the UAVs follow a decoupled dynamics in the $x$ and $y$ directions, and therefore we implement central and local monitors independently for the $x$- and $y$-direction dynamics based on the discretized models of \eqref{eq:cl_sys_scalar}, \eqref{eq:obs_cent}, and \eqref{eq:obs_decent} with the sampling time $T_{\rm s} = 0.02\ \si{sec}$. 

In the following, we present the results of attack detection through the centralized detection scheme with central (global) hypothesis testing \eqref{eq:hypotheses_global} and the central monitor \eqref{eq:obs_cent} as well as through the decentralized detection scheme with local hypothesis testing \eqref{eq:hypotheses_local} and the local monitors of UAVs $1$ and $3$. 
The procedure of the local hypothesis testing is presented in Algorithm \ref{alg:detection_alg_local} and that of the central hypothesis testing is presented in Algorithm \ref{alg:detection_alg_global}. 
\begin{figure}[h]
  \subfloat[UAVs' position in the $x\!-\!y$ plane. The central monitor $\Sigma^{\s \Mc}_{\s \Oc}$ detects the ZDA at $t = 5.6\ \si{sec}$ and ends the experiment.\label{fig:ZDA_sw_XYplane_new}]{\small
    \includegraphics[width=.98\linewidth]{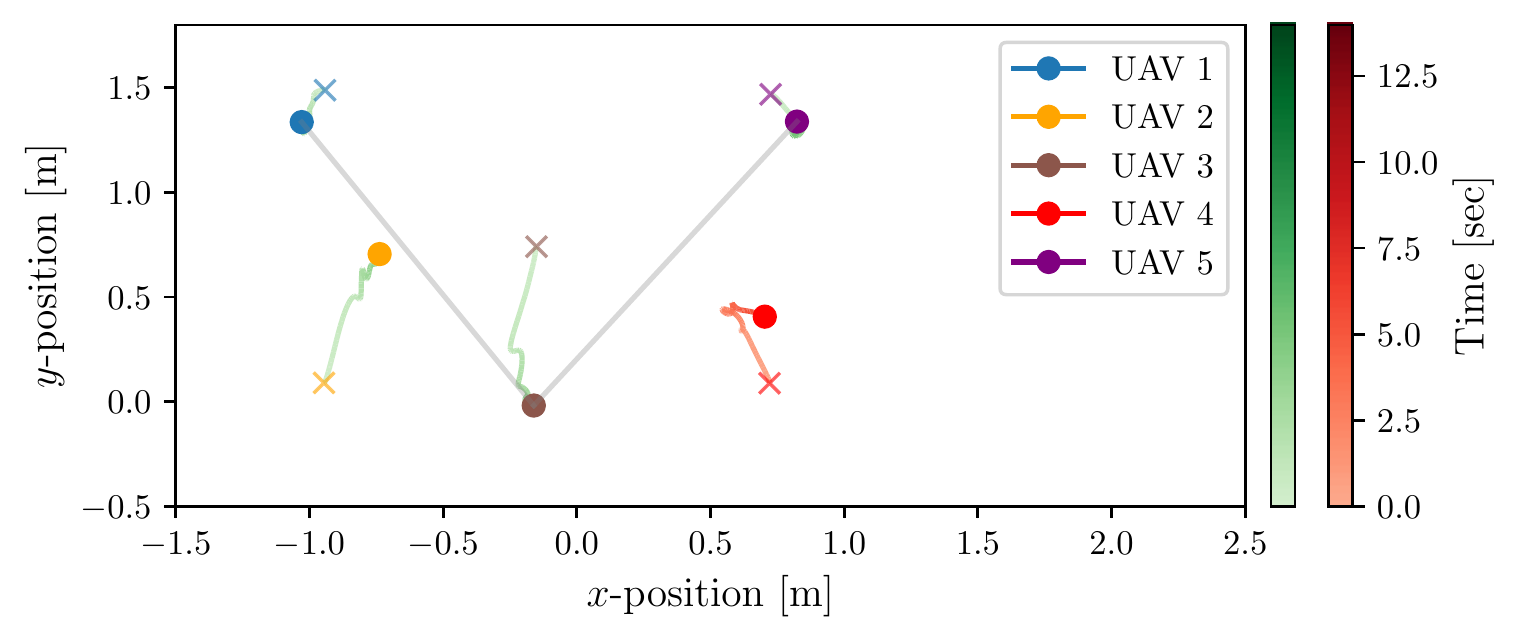}
  }
  \\
 \subfloat[Residuals of local monitor $\Sigma^{\, 1}_{\s \Oc}$ run on UAV $1$. The stealthy ZDA is locally detected at $t = 5.08 \ \si{sec} $. \label{fig:ZDA_sw_local_residual_1_XY_small}]{\small
    \includegraphics[width=.98\linewidth]{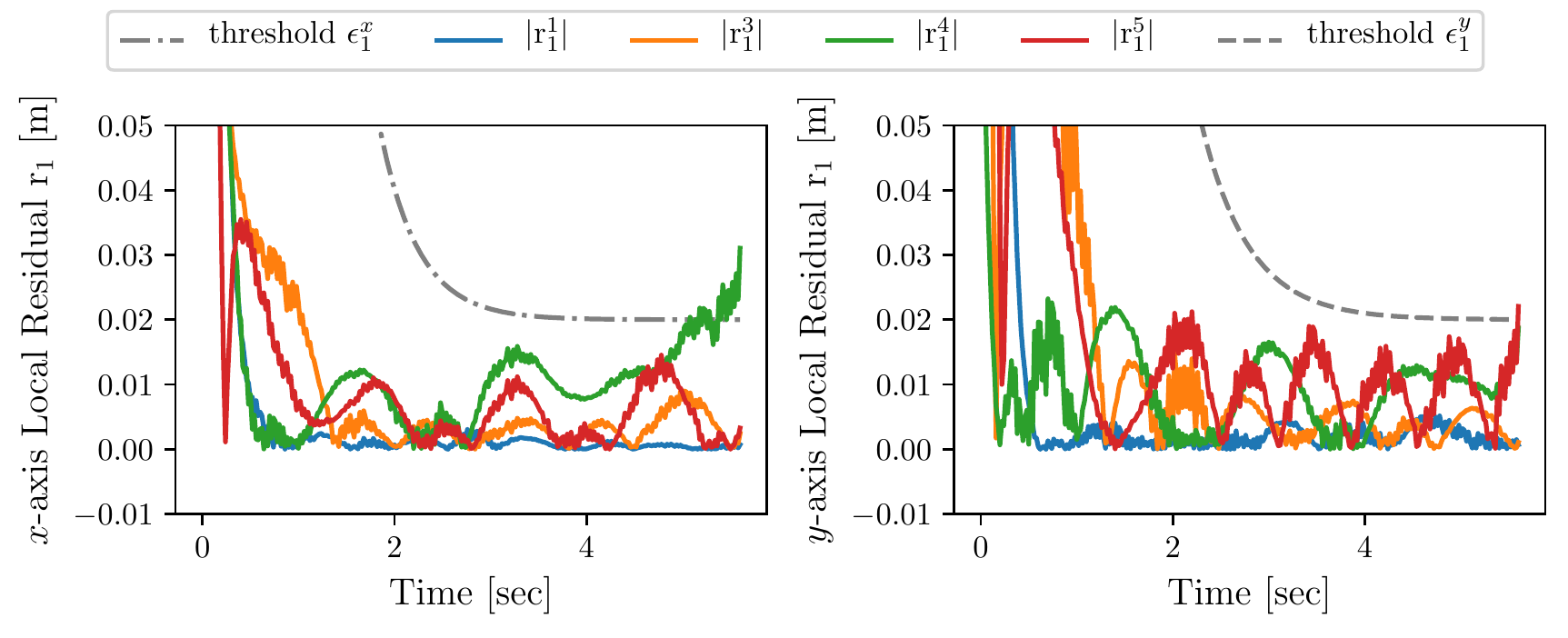}
  }
  \\
  \subfloat[Residuals of central monitor $\Sigma^{\s \Mc}_{\s \Oc}$ run on the control center. The stealthy ZDA is detected at $t = 5.6\ \si{sec} $ using Algorithm \ref{alg:detection_alg_global}. \label{fig:ZDA_sw_central_residual_XY_small}]{\small
    \includegraphics[width=.98\linewidth]{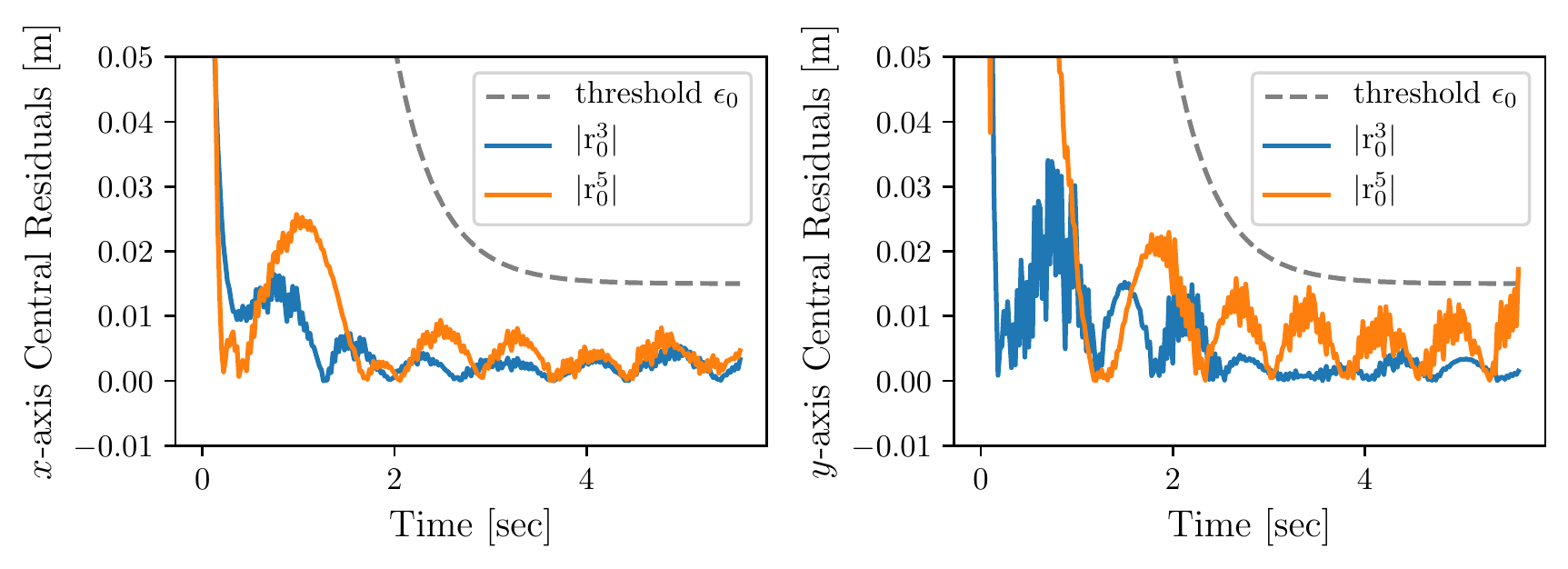}
  }
  \caption{Experiment 2: ZDA on UAVs $1,4,5$ and topology switching from mode $1$ to $3$, which is triggered by local monitor $\Sigma^{\, 1}_{\s \Oc}$ at $t = 5.08\ \si{sec} $. (a) UAVs' position in the $x\!-\!y$ plane with the same annotations as in Figure \ref{fig:ZDA_neutral_sw_position_XY_plane}. (b)-(c) The residuals of local monitor $\Sigma^{\, 1}_{\s \Oc}$ and central monitor $\Sigma^{\s \Mc}_{\s \Oc}$ with the same annotations as in Figures \ref{fig:ZDA_neutral_sw_local_residual_1_XY_small} and \ref{fig:ZDA_neutral_sw_central_residual_XY_small}, respectively.}
  \label{fig:ZDA_sw}
\vspace{-1ex}
\end{figure}
\begin{figure*}[t]
  \subfloat[UAVs' position in the $x\!-\!y$ plane.\label{fig:covert_ramp_XYplane_new}]{\small
    \includegraphics[width=.49\linewidth]{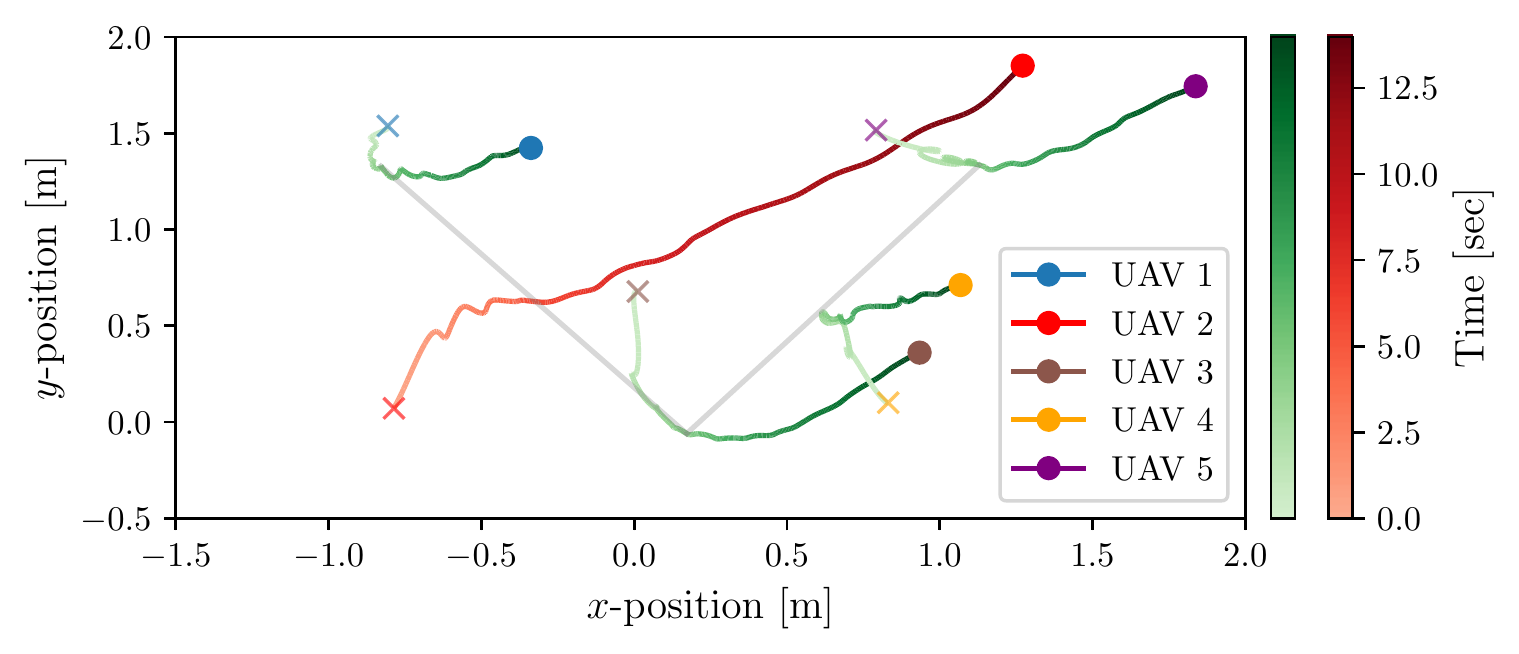}
  }
  \ \
 \subfloat[Realization of stealthiness condition \eqref{eq:stealthy} using $\us$ in \eqref{eq:covert_us} with the starting time $t_{\rm a}=5\ \si{\sec}$.
\label{fig:covert_ramp_measurment_comparison_X}]{\small
    \includegraphics[width=.49\linewidth]{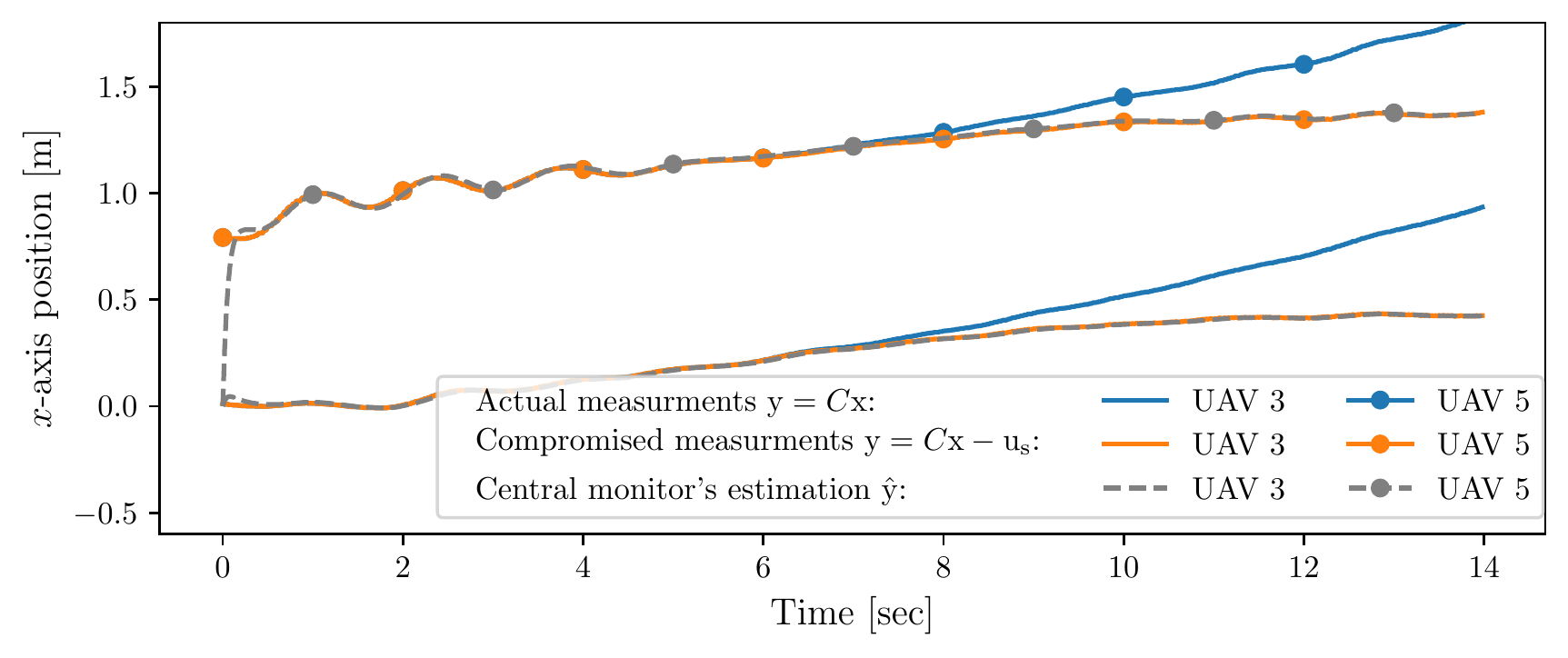}
  }
  \\ \vspace{1ex}
  \subfloat[Residuals of local monitor $\Sigma^{\, 3}_{\s \Oc}$ run on UAV $1$. The stealthy ZDA is detected at $t = 6.4\ \si{sec} $ using Algorithm \ref{alg:detection_alg_local}. \label{fig:covert_ramp_local_residual_3_XY_small}]{\small
    \includegraphics[width=.49\linewidth]{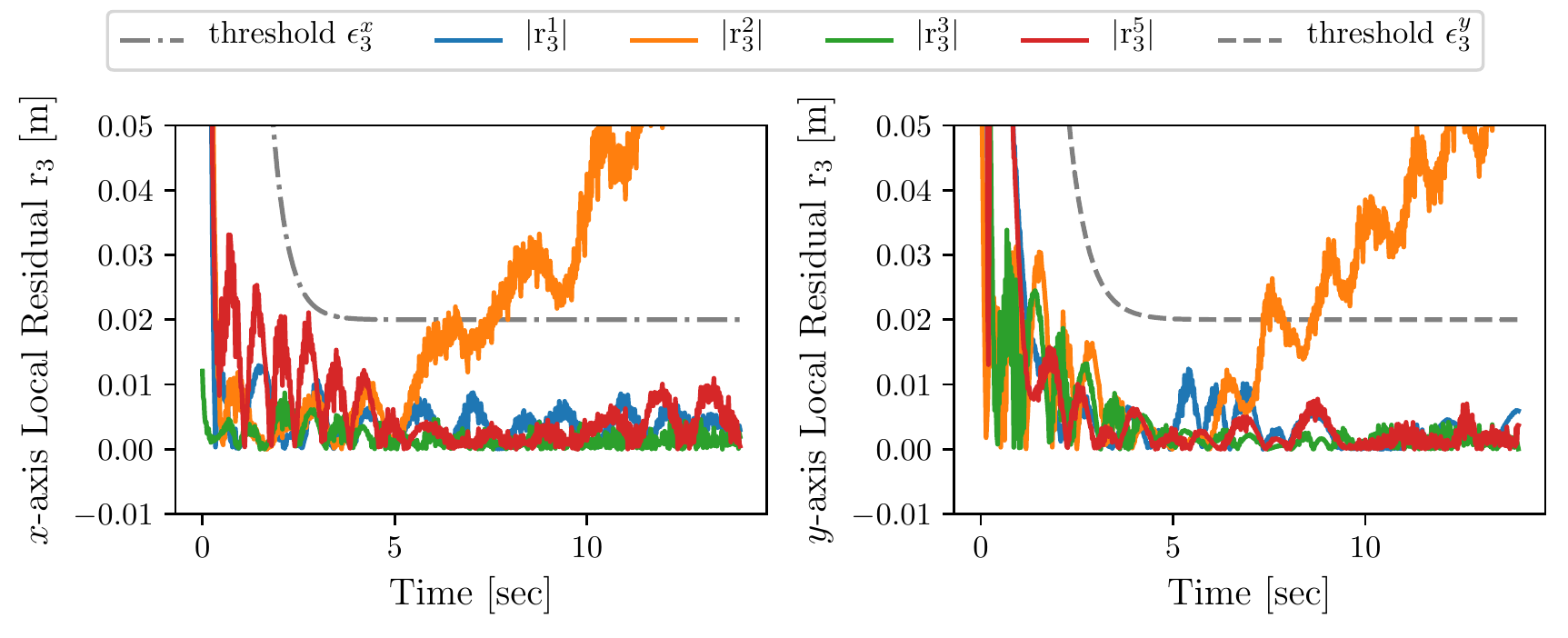}
  }
  \ \
 \subfloat[Residuals of central monitor $\Sigma^{\s \Mc}_{\s \Oc}$ run on the control center. \label{fig:covert_ramp_central_residual_XY_small}]{\small
    \includegraphics[width=.49\linewidth]{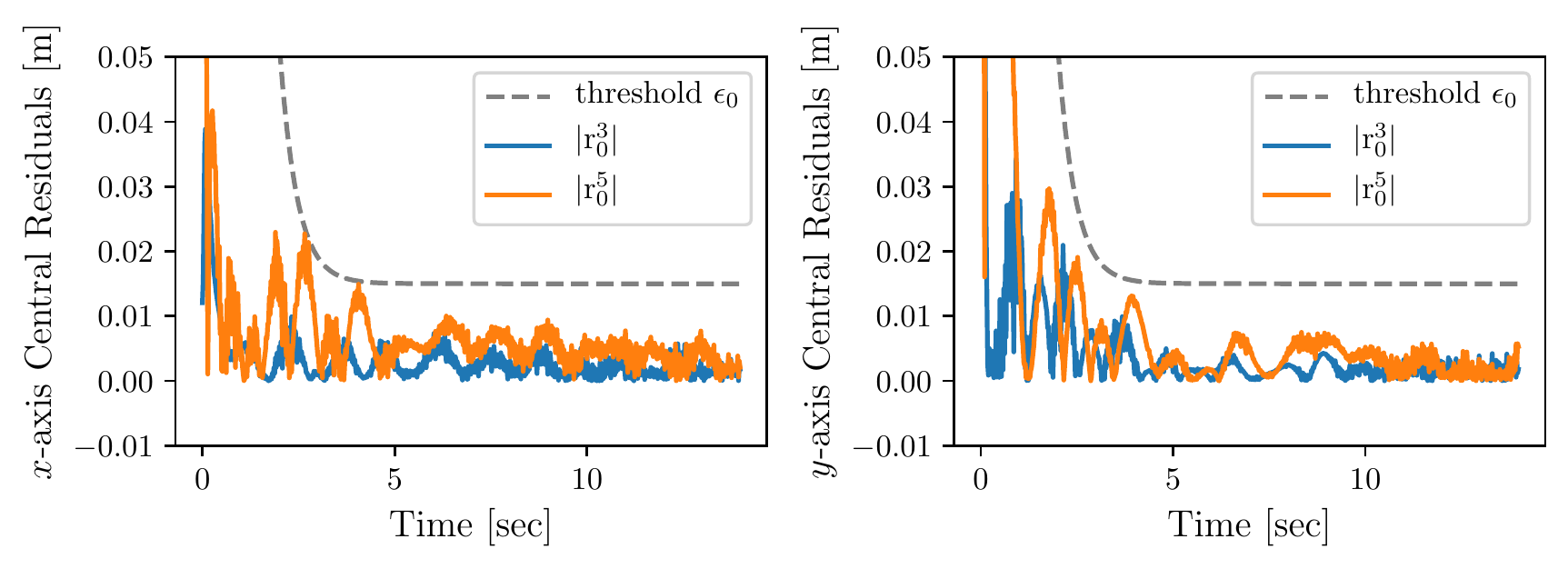}
  }
  \caption{Experiment 3: covert attack on UAV $2$ and topology switching from mode $1$ to $2$, which is triggered by local monitor $\Sigma^{\, 1}_{\s \Oc}$ at $t = 6.4 \ \si{sec} $. (a) UAVs' position in the $x\!-\!y$ plane with the same annotations as in Figure \ref{fig:ZDA_neutral_sw_position_XY_plane}, except the gray lines that visualize the V-shape formation achieved by the UAVs at $t_{\rm a} = 5\ \si{sec}$, the starting time of the covert attack. (b) The effect of measurement alteration using sensory attack $\us$ starting at $t_{\rm a} = 5\ \si{sec}$. (c)-(d) The residuals of local monitor $\Sigma^{\, 3}_{\s \Oc}$ and central monitor $\Sigma^{\s \Mc}_{\s \Oc}$ with the same annotations as in Figures \ref{fig:ZDA_neutral_sw_local_residual_1_XY_small} and \ref{fig:ZDA_neutral_sw_central_residual_XY_small}, respectively.}
  \label{fig:exp_3_covert_ramp}
\vspace{-1ex}
\end{figure*}
\subsubsection{Stealthy zero-dynamics attack} We conducted two experiments evaluating the effectiveness of the central and local monitors in detection of stealthy zero-dynamics attack (ZDA). 
In the first experiment, UAVs $1,\ 4,$ and $5$ are compromised such that their control channels are subject to the discretized version of ZDA signals in \eqref{eq:ZDA_signal} as $\ua = \col \paren{ {\ub}_{{\rm a}_{i}} }_{i \in \Ac} $, $\Ac =\{1,4,5\} $ with ${\ub}_{{\rm a}_{i}} = [{\rm u}^x_{{\rm a}_i}(0)e^{\lambda^x_{o}(k T_{\rm s})}\ \ {\rm u}^y_{{\rm a}_i}(0)e^{\lambda^y_{o}(k T_{\rm s})}\ ]^{\top}$, $\lambda^x_{o} = \lambda^y_{o}  = 0.5,\ k \in \intgnonneg,\ {\rm u}^x_{{\rm a}}(0) =[{\rm u}^x_{{\rm a}_1}(0) \ {\rm u}^x_{{\rm a}_4}(0) \ {\rm u}^x_{{\rm a}_5}(0) ]^{\s \top}  = [-2.34{\rm x}^{\rm a}_{4}(0)\  10.24{\rm x}^{\rm a}_{4}(0)\  {-2.34}{\rm x}^{\rm a}_{4}(0)]^{\s \top},$ $\! {\rm u}^y_{{\rm a}}(0)=-0.7{\rm u}^x_{{\rm a}}(0),\ {\rm x}^{\rm a}_{4}(0)=0.0086 $, and the starting time $t=(k T_{\rm s})=0,\ k = 0$. The network-level measurements \eqref{eq:measurments} with $\Mc_{\rm p} = \{3,5\}$ and $\Mc_{\rm v} = \emptyset $, on the other hand, are not subject to sensory attacks that is $\us = \boldsymbol{0}$. 

In the first experiment, as shown in Figure \ref{fig:ZDA_neutral_sw_position_XY_plane}, all the UAVs start from some initial positions and coordinate to achieve the desired formation while the stealthy ZDA steers UAV $4$ away from its desired configuration that meets \eqref{eq:formation_consensus}. This effect has been illustrated in Figure \ref{fig:ZDA_neutral_sw_relative_disp_y_axis} showing the relative positions of the UAVs as well as their desired values in the $y$ direction over time. It is necessary to note that UAV $4$ hits the safety net enclosing the indoor flight area at $t \approx 9.8\ \si{sec}$. 

In terms of attack detection, Figures \ref{fig:ZDA_neutral_sw_local_residual_1_XY_small} and \ref{fig:ZDA_neutral_sw_local_residual_3_XY_small} show the residuals of local monitors $\Sigma^{\, i}_{\s \Oc}\text{'s},\ i \in \{1,3\} $ in \eqref{eq:obs_decent} for UAVs $1$ and $3$, respectively. Also, Figure \ref{fig:ZDA_neutral_sw_central_residual_XY_small} shows the residuals of the central monitor $\Sigma_{\s \Oc}^{\s \Mc}$ in \eqref{eq:obs_cent} available in the control center. One can verify that the local monitor of UAV $1$, $\Sigma^{\, 1}_{\s \Oc} $, running Algorithm \ref{alg:detection_alg_local}, has detected the stealthy ZDA in a timely manner ($t = 3.22\ \si{sec} $) that is before UAV $4$ collides with the safety net of the flight area at $t \approx 9.8 \ \si{sec} $. However, the ZDA remains stealthy in the residuals of the UAV $3$'s local monitor, $\Sigma^{\, 3}_{\s \Oc} $, and those of the central monitor running Algorithm \ref{alg:detection_alg_global}, regardless of the switch in the inter-UAV's communication topology from mode 1 to mode 4 (see Figure \ref{fig:graphs}) that is triggered by the local monitor $\Sigma^{\, 1}_{\s \Oc} $ at $t = 3.22\ \si{sec} $. This is due to the fact that switching from mode $1$ to mode $4$ does not meet the necessary conditions required for a topology switching to render stealhy attacks detectable for the central monitor in \eqref{eq:obs_cent}. The details of such conditions have been studied in \cite[Th. III.3]{9683427}.

In the second experiment, with the results shown in Figure \ref{fig:ZDA_sw}, the UAVs are under the same ZDA as in the first experiment expect ${\rm x}^{\rm a}_{4}(0)=0.012$. The local monitor $\Sigma^{\, 1}_{\s \Oc}$ successfully detects the ZDA at $t = 5.08\ \si{sec}$ (see Figure \ref{fig:ZDA_sw_local_residual_1_XY_small}) and then triggers a switch in the UAVs' communication from mode $1$ to mode $3$ (cf. Figure \ref{fig:graphs}) that as opposed to Experiment 1, this topology switching results in detection of stealthy ZDA by the central monitor $\Sigma^{\s \Mc}_{\s \Oc}$ at $t = 5.6\ \si{sec}$ (see Figure \ref{fig:ZDA_sw_central_residual_XY_small}). 

It is necessary to note that any topology switching in the inter-UAVs communications results in a discrepancy between the actual dynamics of networked UAVs and its nominal counterpart that is used by the attacker to design stealthy attacks. Yet, the model discrepancy in Experiment 1 did not interfere with the stealthiness of ZDA in the central monitor's residuals while it renders ZDA detectable in the central monitor's residuals in Experiment 2.
These results indicate that not only zero-dynamics  attacks (ZDA) can be implemented in real-time on networked UAVs with partial measurements, they also can remain stealthy regardless of switches in the inter-UAVs' communication topology. Theoretical results to detect stealthy ZDA through topology switching in networked systems with full-state measurements and with partial measurements can be found, respectively, in  \cite{mao2020novel} and \cite{9683427}.
\subsubsection{Covert attack} 
Similar to the ZDA case, we evaluated the detection of covert attack on networked UAVs subject to topology switching by using the local monitors $\Sigma^{\, 1}_{\s \Oc}$ and $\Sigma^{\, 3}_{\s \Oc}$, and the central monitor $\Sigma^{\s \Mc}_{\s \Oc}$. In this experiment, a covert attack, $\mb{u}_{a_ i},\ i \in \Ac = \{2\} $, in the form of a ramp signal with a slope of $3\ \si{Deg.}$, as the roll and pitch angles' perturbation, and the starting time of $t_{\rm a} = 5\ \si{sec}$ is injected through the control channel of UAV $2$. 
The covert attack's effect on the UAVs' formation is shown in Figure \ref{fig:covert_ramp_XYplane_new}. As illustrated, all of the UAVs have deviated from their desired formation configuration that meets \eqref{eq:formation_consensus}. The effect of this deviation/perturbation on the measurements $\mb{y}$ in \eqref{eq:measurments} is simultaneously canceled out by implementing the discretized version of the sensory attack $\us$ given in \eqref{eq:covert_us}. Figure \ref{fig:covert_ramp_measurment_comparison_X} illustrates how the alteration of actual  measurement $\mb{y}$ of the monitored UAV $3$ using the sensory attack $\us$ gives rise to a false state estimation by the central monitor $\Sigma^{\s \Mc}_{\s \Oc}$, rendering the injected attack covert in the central residuals. 
The local monitors, however, are not subject to such alterations and thus are capable of detecting the covert attack in a timely manner as shown in Figure \ref{fig:covert_ramp_local_residual_3_XY_small} for the local monitor $\Sigma^{\, 3}_{\s \Oc}$ of UAV $3$. We note that the local monitor $\Sigma^{\, 3}_{\s \Oc}$ triggers a topology switch from mode $1$ to mode $2$ (cf. Figure \ref{fig:graphs}) at $t = 6.4\ \si{sec} $ to make the covert attack detectable in the residuals of the central monitor $\Sigma^{\s \Mc}_{\s \Oc}$. However, the attack remains stealthy in the central residuals, shown in Figure \ref{fig:covert_ramp_central_residual_XY_small}, regardless of topology switching. The results, consistent with those in the ZDA case, shows the outperformance of the decentralized detection scheme (Algorithm \ref{alg:detection_alg_local}) over the centralized detection scheme (Algorithm \ref{alg:detection_alg_global}). 
It is worth mentioning that one can leverage a larger number of switching communication links on which the centralized monitor relies to improve the performance of the centralized detection scheme. However, this solution raises other challenges such as switching-induced unobservability as well as communication overhead. In the case of the decentralized detection scheme that relies on the network model of UAVs, scalability is a concern for larger teams of UAVs, for which clustering-based solutions such as the one in \cite{9683427} can be applied.

\section{Conclusions}\label{Sec:conclusion}
In this paper, we studied the detection of sophisticated stealthy attacks including zero-dynamics and covert attacks on networked UAVs with switching communication topology in formation control settings. Centralized and decentralized attack detection schemes were developed to detect the stealthy attacks. 
We performed indoor flight tests using quadrotor UAVs to evaluate the effectiveness of our proposed methods. The experiments showed feasibility of stealthy attacks on UAV networks which was previously investigated only in theory or simulations. The experiments also validated our theoretical developments and the effectiveness of the attack detection schemes proposed in this paper. The experimental results demonstrated the importance of timely detection of stealthy attacks and trade-offs existing in centralized vs. decentralized attack detection. Future work will be devoted to extending the results to more complex cooperative settings where we will incorporate cooperative attack mitigation and mission adaptation.
%
\section*{Acknowledgment}
This research is partially supported by the National Science Foundation (award no. 2137753).

\bibliographystyle{IEEEtran}
\bibliography{IEEEabrv,references}

\end{document}